\documentclass[journal]{IEEEtran}
\usepackage{amsmath,amssymb,amsthm}
\usepackage{booktabs}
\usepackage{graphicx}
\usepackage{cite}
\usepackage{tikz}
\usetikzlibrary{arrows.meta,positioning,calc}
\usepackage{algorithm}
\usepackage{algorithmic}
\usepackage[colorlinks=true,linkcolor=blue,citecolor=blue]{hyperref}

\newtheorem{theorem}{Theorem}
\newtheorem{proposition}{Proposition}
\newtheorem{assumption}{Assumption}
\newtheorem{remark}{Remark}
\newcommand{\bL}{\mathbf{L}}
\newcommand{\bM}{\mathbf{M}}
\newcommand{\bPhi}{\boldsymbol{\Phi}}
\newcommand{\bgam}{\boldsymbol{\gamma}}
\newcommand{\rL}{\rho(\bL)}

\begin{document}

\title{Output-Only Identification and Spectral Monitoring of Coupled
Feedback Networks with Known Time-Varying Actuation}

\author{Jihwan~Woo%
\thanks{J. Woo is with Amazon Web Services (e-mail available on request).
The views expressed are those of the author and do not represent those
of Amazon Web Services or its affiliates. An application-focused
companion paper on the leveraged-fund episode is available as an arXiv
preprint; this paper develops the general identification theory.}}

\markboth{Working paper --- August 2026}%
{Woo: Output-Only Identification of Coupled Feedback Networks}

\maketitle

\begin{abstract}
Coupled feedback networks are often monitored channel by channel even
though cross-channel paths alter both stability margins and transmitted
disturbances. We study identification of a structured feedback matrix
$\bL_t=\bPhi\,\mathrm{diag}(\bgam_t)$ in an \emph{output-only}
setting: no commanded, probing, or reference input exists---only
temporally separated outputs and the scheduling gains $\bgam_t$ are
observed, while the coupling response $\bPhi$ and the
clearing-window inputs are not. Identification then rests on two
sources jointly: the persistent excitation of the observed
pre-window output, and two structural features that separate
coupling from confounds: the known time
variation of the gains, which modulates the closed-loop response in a
predictable pattern, and a partial-reversal moment, by which a known
fraction of transient displacement is corrected in a subsequent
observation window. We give a hierarchy of results: regular local
identification in the exact fixed-point model under a Jacobian rank
condition on the
gain regimes; a first-order interaction estimator whose
identification strength is the minimum eigenvalue of the residualized
interaction information matrix (provably unidentified when gains are
constant); and a characterization of the estimand as a resolvent
sensitivity---directly the right object for screening transmitted
disturbances, and a first-stage input to spectral-margin recovery.
We establish $\sqrt{T}$ asymptotics for the first-order estimator.
Under locally full-dimensional positive gain variation and an
outgoing-path nondegeneracy condition, a cross-identification
theorem maps each varying gain to point-identified resolvent rows
and columns. Pointwise bootstrap results cover both a smooth
two-stage inversion and a null-imposed outward-score test for hard
SCA; the latter has $6.7\%$ boundary size at a nominal $5\%$ and
power $25\%/85\%$ at local alternatives $0.33/0.36$. A
margin-switched construction reduces uniform spectral inference to
one open input---a degeneracy-valid margin bound---and drift and
near-degenerate null designs reject on no tested path. A log-coordinate convex
geometry yields a hard-constrained, SVD-initialized backtracking SCA
with local branch-capture guarantees; against matched box-TRF it
removes both catastrophic branches in 25 paths and reduces spectral
bias from $+0.100$ to $+0.043$. The identification guarantees require the clean pre-window
moment under stated exogeneity and orthogonality assumptions; the
full-period-regressor case is characterized as a contaminated
estimand.
Simulations verify sharpness of the rank condition, quantify the
conditioning effect of correlated gain paths, and show that
level-regression benchmarks fail under sign-adverse confounding,
while a deliberately mismatched two-regime variance-ratio
construction in the spirit of
heteroskedasticity-based identification is unstable on the same
data, whose structural shock variances are
constant by design. A case study on leveraged-fund
rebalancing feedback, where daily fund disclosures play the role of
the known gains, illustrates the method on real data.
\end{abstract}



\section{Introduction}\label{sec:intro}

\IEEEPARstart{M}{any} engineered and economic systems are collections
of feedback loops that share a medium: power grids with
demand-response contracts, recommendation platforms whose rankings
feed back into the behavior being ranked, advertising exchanges where
bidding algorithms respond to prices they jointly set, and financial
markets hosting leveraged products whose mandated rebalancing trades
move the prices the mandates are written on. In all of these, each
loop is typically audited in isolation---its own gain, its own
stability margin---while the loops are coupled through the shared
medium. For a coupled system with elementwise-nonnegative gain matrix
$\bL$, the spectral radius satisfies $\rL\ge\max_i \ell_{ii}$: the
diagonal audit is guaranteed optimistic, and the gap can be
substantial when off-diagonal cycles are strong. Off the diagonal, even
one-way coupling transmits displacement from loop to loop through the
resolvent $(\mathbf{I}-\bL)^{-1}$, a channel invisible to every
per-loop statistic.

Estimating $\bL$ is an identification problem that standard tools do
not directly cover. The system operates in closed loop, but there is no
reference signal or probing input, so classical closed-loop
identification \cite{ljung1999} does not apply. The node signals form
a cyclic dynamic network, but external excitation---whose presence
and location drive identifiability in the dynamic-network literature
\cite{weerts2018,cheng2023}---is absent. The loop gains are
parameter-varying, but the input signal whose observation
linear-parameter-varying (LPV) identification requires
\cite{toth2010,cox2021} is latent. Output-only network reconstruction
without further structure is subject to fundamental limits
\cite{materassi2012,angulo2017}. And identification through
heteroskedasticity \cite{rigobon2003} exploits regime changes in
unobserved shock variances under a constant structure---the exact
dual of our setting, in which the structure varies observably while
shock variances need not.

What rescues \emph{structural} identifiability here is a feature
these frameworks do not use: \emph{the actuation gains are known}.
(Signal excitation alone cannot do it: a static confound and a
static coupling are observationally merged until something known
moves.) In the applications
above, the gain of each loop is disclosed or contractually
determined---fund assets and leverage multiples, contracted
demand-response capacity, published platform parameters---and it
varies over time for reasons unrelated to the disturbances. The known
time-varying gain acts as a natural modulating sequence: it changes
the closed-loop response in a predictable pattern, so that the
response to the \emph{same} latent input differs across gain regimes
in a way that is informative about the unknown coupling. This is the
feedback-network analogue of a classical signal-processing idea:
transmitter-induced cyclostationarity, where a known periodic
modulation renders blind FIR channel identification well posed
\cite{tsatsanis1997}. Our modulation is not designed and not
periodic---it is disclosed by the system itself---and the object it
identifies is not an open-loop channel but a closed-loop coupling
matrix inside a fixed point. A second structural feature sharpens the
moment: in our motivating application the displacement caused by the
feedback is \emph{transient}, and a known fraction $\theta$ of it
reverts in a subsequent observation window, so the reversal itself is
a signed, gain-scaled signature of the feedback path.

\emph{Contributions.} (i) We formalize output-only identification of
$\bL_t=\bPhi\,\mathrm{diag}(\bgam_t)$ from outputs and known gains
$\bgam_t$, and derive an interaction estimator based on the
reversal moment (Section~\ref{sec:problem}--\ref{sec:ident}).
(ii) We give identifiability conditions and prove their sharpness
at two levels: in the linearized coefficient map, per-entry
identification requires time variation of the corresponding gain
path and constant gains leave a continuum of observationally
equivalent (coupling, confound) pairs (Theorem~\ref{thm:ident});
regular local identification in the exact fixed-point model is
governed by the rank of a stacked
Jacobian over gain regimes (Theorem~\ref{thm:nonlinear}). (iii) We characterize the estimand as a
\emph{resolvent sensitivity} rather than the direct coupling: the
detector responds to network-reachable paths---directly the
desirable object for transmission screening, and a first-stage
input to spectral recovery---and we separate direct-edge and
spectral recovery as a second-stage inverse problem
(Section~\ref{sec:estimand}). (iv) We develop the estimation and
inference layer: $\sqrt{T}$ asymptotics for the first-order
estimator, a cross-identification theorem that maps each varying
gain to point-identified resolvent rows and columns under local
support and outgoing-path conditions, with a non-iterative SVD
estimator, bootstrap validity under consistent selection, and a
log-coordinate geometry yielding a hard-constrained SVD-initialized
SCA, with critical-point convergence and local branch/statistical
stability guarantees
(Theorems~\ref{thm:cross}--\ref{thm:bootstrap},
Proposition~\ref{prop:mmgeom}; Algorithm~\ref{alg:sca}). (v) Simulations verify the sharpness
of the rank condition and quantify failure modes of level-regression
and heteroskedasticity-based benchmarks (Section~\ref{sec:sims}).
(vi) A real-data case study on leveraged-fund rebalancing feedback
illustrates the pipeline end to end (Section~\ref{sec:case}); the
companion empirical paper \cite{woo2026arxiv} develops that
application in full, while the present paper supplies the
identification and monitoring theory.

Table~\ref{tab:delta} positions the setting against its nearest
neighbors. The precise claim is structural: to our knowledge, no
existing framework treats identification of an algebraic feedback
coupling from outputs alone when the identifying variation is a
known, column-wise, time-varying gain and the informative moment is
a signed partial reversal; individual ingredients appear separately
in the table's rows.

\begin{table}[t]
\caption{Positioning against neighboring identification frameworks.}
\vspace{2pt}
\label{tab:delta}
\centering
\footnotesize
\begin{tabular}{@{}p{2.1cm}p{1.9cm}p{1.9cm}p{1.6cm}@{}}
\toprule
Framework & Observed & Excitation & Output \\
\midrule
Closed-loop ID \cite{ljung1999} & $u,y$, reference & external reference & plant TF \\
Dynamic networks \cite{weerts2018} & nodes + excitation & external signals & module TFs \\
LPV-ID \cite{toth2010} & $u,y$, scheduling & input $u$ & $A(p),B(p)$ \\
Heteroskedasticity \cite{rigobon2003} & $y$ & changes in shock covariance & constant $A_0$ \\
MIC$^{\;\dagger}$ blind ID \cite{tsatsanis1997} & $y$, periodic precoder & designed modulation & FIR channel \\
Output-only recon. \cite{materassi2012,angulo2017} & $y$ only & none; limits apply & topology (under assumptions) \\
\textbf{This paper} & $y$, gains $\bgam_t$ & \textbf{known time-varying gains} & $\bPhi$ (1st-order), screening alarms$^{*}$ \\
\bottomrule
\end{tabular}

\vspace{2pt}
\parbox{\columnwidth}{\scriptsize $^{\dagger}$MIC = modulation-induced cyclostationarity.
$^{*}$Exact $\bPhi$ under regime-specific clean moments
(Theorem~2); spectral recovery via a heuristic second stage with
calibrated inference open (Section~V).}
\end{table}

\section{Related Work}\label{sec:related}

Four literatures approach the identification of interconnected
dynamical systems, and each leaves the present setting uncovered for
a structural reason.

\emph{Closed-loop and network identification.} Classical closed-loop
identification resolves the circularity of feedback through an
external reference or probing signal \cite{ljung1999,forssell1999};
prediction-error methods for dynamic networks generalize this to
structured interconnections \cite{vandenhof2013}, and a mature theory
now characterizes network identifiability by the presence,
location, and rank of \emph{external} excitation and disturbance
signals \cite{weerts2018,hendrickx2019,cheng2023}. Topology-oriented
variants reconstruct the interconnection structure from node spectra
\cite{goncalves2008,materassi2010,materassi2012}. Two obstructions
separate our problem from this line. First, no commanded excitation
exists---no reference, no probing, no exogenously excited node; the
identifiability currency of the network literature is unavailable.
Second, the object is a static-within-window coupling embedded in an
algebraic fixed point \eqref{eq:fp}, with a cycle at its core,
whereas sharp partial-excitation results are largely acyclic
\cite{cheng2023}. To be precise about what is and is not missing:
the pre-window output $\mathbf{r}_1$ is observed and persistently
exciting, so it supplies the forcing any regression needs; what is
absent is any \emph{commanded or probing} signal, and what takes its
place as the source of \emph{structural} identification---separating
coupling from gain-invariant confounds---is the known time variation
of the actuation gains, which has no counterpart in this theory.

\emph{LPV identification.} If the gain path $\bgam_t$ is read as a
scheduling signal, our model is a linear parameter-varying system
with known scheduling \cite{toth2010}. LPV prediction-error and
subspace methods \cite{laurain2010,cox2021} estimate the full
parameter-varying dynamics, including closed-loop variants---but all
require the input signal to be measured; identifiability is inherited
from input excitation, exactly as in the LTI case. Recent work treats
the \emph{scheduling} as latent and learns it \cite{verhoek2022}. Our
problem inverts both premises: the input is latent, the scheduling is
known, and the question is whether known scheduling variation alone
can replace input observation. To our knowledge this corner---output-only LPV identification with
scheduling entering only the feedback path---appears unoccupied,
though we state this as a literature claim open to correction
rather than a theorem.

\emph{Identification through heteroskedasticity.} Econometrics
identifies simultaneous systems from regime changes in the
\emph{variances} of latent shocks under a constant structural matrix
\cite{rigobon2003,sentana2001,lanne2008,lewis2021,kilian2017}. This
is the closest conceptual relative: both approaches buy
identification from time variation rather than exclusion restrictions
or instruments. But the two are duals, not variants. There, the
structure $\mathbf{A}_0$ is constant and the shock covariance moves,
unobserved, estimated as regimes; here, the structure
$\bL_t=\bPhi\,\mathrm{diag}(\bgam_t)$ moves \emph{observably} while
shock covariances may be constant. The duality is not cosmetic:
Section~\ref{sec:sims} includes a deliberately mismatched two-regime
variance-ratio benchmark on data generated by our model---with
structural shock variances constant by design, its variance-ratio
denominator is unstable (reduced-form variances do move with the
gains, but not in the way that construction requires), which
illustrates the model mismatch rather than numerically refuting
heteroskedasticity identification---while our estimator is
undefined on data generated by
theirs (constant gains violate the rank condition of
Theorem~\ref{thm:ident}). In the two stylized designs compared
here, each method's identifying variation is the
other's maintained constancy.

\emph{Adjacent model classes.} Three further neighbors deserve
explicit contrast. Bilinear system identification
\cite{fnaiech1987} treats products of states and \emph{inputs}, so
its excitation is again a measured input, absent here. The
varying-coefficient regression literature \cite{hastie1993}
supplies exactly the reduced-form skeleton of our
\eqref{eq:regression}---coefficients driven by an observed
effect modifier---but stops at the conditional mean: it has no
feedback fixed point, no reversal moment, and therefore no map from
fitted coefficients to a structural coupling matrix or its spectral
radius; our contribution is precisely that structural layer.
Output-only (operational) modal analysis \cite{peeters1999}
identifies modal structure from ambient responses, but under an
unobserved stationary white excitation and time-invariant
dynamics---the two assumptions our setting replaces with a known,
nonstationary gain path.

\emph{Known modulation in signal processing.} The principle that a
known time-varying modulation renders a blind problem well posed is
classical in communications: transmitter-induced cyclostationarity
and periodic precoding enable blind FIR channel identification from
second-order statistics \cite{tsatsanis1997,serpedin1998,gardner1991}.
We import the principle and change every ingredient: the modulation
$\bgam_t$ is disclosed by the system rather than designed, aperiodic
rather than cyclostationary, and enters a feedback fixed point rather
than a feedforward channel; identification additionally exploits a
signed reversal moment \eqref{eq:reversal} with no communications
analogue. Conversely, the fundamental limits established for output-only
network reconstruction \cite{materassi2012,angulo2017} delimit what
any method can do without such structure---our contribution can be read as exhibiting one structural
supplement (known gains plus reversal) under which the
limitation dissolves; the individual mathematical devices are
standard, and the novelty claim rests on the model and its use of
disclosed modulation, not on any single tool.

\section{Problem Formulation}\label{sec:problem}

\subsection{Observation model}

There are $n$ channels. At each period $t$ we observe a pre-window
output $\mathbf{r}_{1,t}\in\mathbb{R}^n$, a post-window output
$\mathbf{r}_{\mathrm{on},t}\in\mathbb{R}^n$, and the actuation gains
$\bgam_t\in\mathbb{R}_{+}^n$. The within-period feedback operates in
a clearing window between the two observations. Each channel $j$
places a mandated action proportional to its own gain and its
\emph{full-period} signal $r_{1,j,t}+r_{2,j,t}$, where
$\mathbf{r}_{2,t}$ is the (latent) clearing-window
displacement---the self-reference that creates the fixed point---and actions displace all channels
through a coupling response $\bPhi=[\phi_{ij}]\ge0$:
\begin{equation}
\mathbf{r}_{2,t}
=\bPhi\,\mathrm{diag}(\bgam_t)\,(\mathbf{r}_{1,t}+\mathbf{r}_{2,t})
+\mathbf{v}_t
=\bL_t(\mathbf{r}_{1,t}+\mathbf{r}_{2,t})+\mathbf{v}_t,
\label{eq:selfref}
\end{equation}
where $\mathbf{v}_t$ is latent clearing-window noise. Collecting
$\mathbf{r}_{2,t}$ terms, $(\mathbf{I}-\bL_t)\mathbf{r}_{2,t}
=\bL_t\mathbf{r}_{1,t}+\mathbf{v}_t$, so that whenever
$1\notin\sigma(\bL_t)$,
\begin{equation}
\mathbf{r}_{2,t}=(\mathbf{I}-\bL_t)^{-1}
(\bL_t\mathbf{r}_{1,t}+\mathbf{v}_t).
\label{eq:fp}
\end{equation}
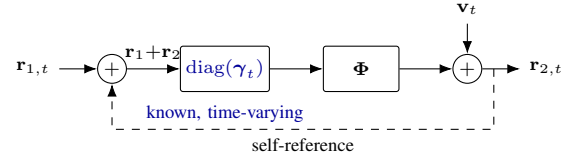
\begin{figure}[t]
\centering
\begin{tikzpicture}[node distance=5mm and 7mm, font=\scriptsize,
 box/.style={draw, rounded corners=1pt, minimum height=6mm,
   minimum width=10mm, inner sep=2pt},
 sum/.style={draw, circle, inner sep=1pt}]
\node[sum] (s1) {$+$};
\node[box, right=of s1, text=blue!60!black] (gam) {$\mathrm{diag}(\bgam_t)$};
\node[box, right=of gam] (phi) {$\bPhi$};
\node[sum, right=of phi] (s2) {$+$};
\node[left=5mm of s1] (r1) {$\mathbf{r}_{1,t}$};
\node[above=4mm of s2] (v) {$\mathbf{v}_t$};
\node[right=5mm of s2] (r2) {$\mathbf{r}_{2,t}$};
\draw[-Latex] (r1) -- (s1);
\draw[-Latex] (s1) -- node[above]{$\mathbf{r}_1{+}\mathbf{r}_2$} (gam);
\draw[-Latex] (gam) -- (phi);
\draw[-Latex] (phi) -- (s2);
\draw[-Latex] (v) -- (s2);
\draw[-Latex] (s2) -- (r2);
\draw[-Latex, dashed] ($(s2)+(3.5mm,0)$) -- ++(0,-8mm) -|
 node[pos=0.25, below]{self-reference} (s1.south);
\node[below=0.5mm of gam, text=blue!60!black]
 {\scriptsize known, time-varying};
\end{tikzpicture}
\caption{The coupled feedback fixed point \eqref{eq:selfref}.}
\vspace{2pt}
\parbox{\columnwidth}{\footnotesize Note: the actuation gains $\mathrm{diag}(\bgam_t)$
are disclosed---the only observed element of the loop besides the
outputs---and their time variation is the identification resource;
the coupling $\bPhi$ and inputs $\mathbf{v}_t$ are unobserved. The
dashed path is the self-reference that makes $\mathbf{r}_{2,t}$ a
fixed point.}
\label{fig:block}
\end{figure}

We maintain $\rho(\bL_t)<1$, under which the resolvent admits the
Neumann expansion
\begin{equation}
(\mathbf{I}-\bL_t)^{-1}=\sum_{k=0}^{\infty}\bL_t^{k},
\qquad
\bM_t\equiv(\mathbf{I}-\bL_t)^{-1}-\mathbf{I}
=\sum_{k=1}^{\infty}\bL_t^{k},
\label{eq:neumann}
\end{equation}
whose $k$th term represents the $k$th feedback round;
off-diagonal (cross-channel) transmission enters already at $k=1$.
Two identities used repeatedly follow directly from the definition of
$\bM_t$:
\begin{equation}
(\mathbf{I}-\bL_t)^{-1}=\mathbf{I}+\bM_t,
\qquad
(\mathbf{I}-\bL_t)^{-1}\bL_t=\bM_t.
\label{eq:identities}
\end{equation}
The post-window mechanism reverses a known share
$\theta\in(0,1]$ of the displacement and carries a constant confound:
\begin{equation}
\mathbf{r}_{\mathrm{on},t}
=-\theta\,\mathbf{r}_{2,t}+\mathbf{C}\,\mathbf{r}_{1,t}+\mathbf{w}_t.
\label{eq:reversal}
\end{equation}
Unknowns: the coupling $\bPhi\in\mathbb{R}^{n\times n}$ (constant),
the confound $\mathbf{C}$ (constant; lead--lag or common-factor
loadings), and the noise second moments. Known: $\bgam_t$ (disclosed)
and $\theta$ (calibrated externally or normalized; to first order
unknown $\theta$ is a global scale ambiguity on $\bPhi$, while in
the exact model powers of $\bPhi$ enter at different orders---see
the Remark after Theorem~\ref{thm:nonlinear}).

\begin{assumption}[Exogenous gains]\label{as:exo}
$\{\bgam_t\}$ is deterministic or independent of
$\{\mathbf{r}_{1,t},\mathbf{v}_t,\mathbf{w}_t\}$, and $\mathbf{C}$
does not vary with $\bgam_t$.
\end{assumption}

Assumption~\ref{as:exo} is the exclusion restriction: the confound
may be arbitrary but must not track the disclosed gains. Its failure
mode is quantified in Section~\ref{sec:sims} (scenario C).

\begin{assumption}[Stochastic regularity]\label{as:stoch}
$\{\mathbf{r}_{1,t},\mathbf{v}_t,\mathbf{w}_t\}$ is stationary and
ergodic with finite fourth moments;
$\mathbb{E}[\mathbf{v}_t\,|\,\mathbf{r}_{1,t},\bgam_t]=
\mathbb{E}[\mathbf{w}_t\,|\,\mathbf{r}_{1,t},\bgam_t]=\mathbf{0}$;
and the pre-window signal is persistently exciting conditional on
the gains: $\mathbb{E}[\mathbf{r}_{1,t}\mathbf{r}_{1,t}^{\!\top}
\,|\,\bgam_t]\succ0$.
\end{assumption}

Gain variation alone is not sufficient: identification also needs
signal excitation and conditional orthogonality, which
Assumption~\ref{as:stoch} supplies. Theorems~\ref{thm:ident}
and~\ref{thm:nonlinear} are population-level identification results
under Assumptions~\ref{as:exo}--\ref{as:stoch}; the accompanying
finite-sample estimation theory (rates, asymptotic normality under
mixing, efficiency) is part of the UQ program of
Section~\ref{sec:uq} and is stated there as open where it is open.

\subsection{The moment, derived}

Substituting \eqref{eq:fp} into \eqref{eq:reversal} and using
\eqref{eq:identities},
\begin{align}
\mathbf{r}_{2,t}
&=(\mathbf{I}-\bL_t)^{-1}\bL_t\,\mathbf{r}_{1,t}
+(\mathbf{I}-\bL_t)^{-1}\mathbf{v}_t \nonumber\\
&=\bM_t\,\mathbf{r}_{1,t}+(\mathbf{I}+\bM_t)\,\mathbf{v}_t,
\label{eq:r2decomp}
\end{align}
so the post-window output is
\begin{equation}
\mathbf{r}_{\mathrm{on},t}
=(\mathbf{C}-\theta\bM_t)\,\mathbf{r}_{1,t}
-\theta(\mathbf{I}+\bM_t)\,\mathbf{v}_t+\mathbf{w}_t.
\label{eq:ondecomp}
\end{equation}
Under Assumption~\ref{as:stoch}, augmented with
$\mathbb{E}[\mathbf{w}_t\mathbf{v}_t^{\!\top}\,|\,\bgam_t]=\mathbf{0}$
(post-window noise orthogonal to clearing-window noise, needed only
for the full-period regressor below), the population projection of
$\mathbf{r}_{\mathrm{on},t}$ on $\mathbf{r}_{1,t}$ \emph{conditional
on the gain regime} $\bgam_t$ is (writing
$\boldsymbol{\Sigma}_1=\mathbb{E}[\mathbf{r}_{1,t}\mathbf{r}_{1,t}^{\!\top}\,|\,\bgam_t]$,
$\boldsymbol{\Sigma}_v=\mathbb{E}[\mathbf{v}_t\mathbf{v}_t^{\!\top}\,|\,\bgam_t]$,
which by gain independence (Assumption~\ref{as:exo}) coincide with
their unconditional counterparts)
\begin{align}
\mathbf{B}(\bgam_t)
&=\mathbb{E}[\mathbf{r}_{\mathrm{on},t}\mathbf{r}_{1,t}^{\!\top}\,|\,\bgam_t]\,
\mathbb{E}[\mathbf{r}_{1,t}\mathbf{r}_{1,t}^{\!\top}\,|\,\bgam_t]^{-1}
\nonumber\\
&=(\mathbf{C}-\theta\bM_t)\,\boldsymbol{\Sigma}_1\,
\boldsymbol{\Sigma}_1^{-1}
=\mathbf{C}-\theta\,\bM_t.
\label{eq:momentclean}
\end{align}
The reversal moment thus identifies $\bM_t$ up to the known scale
$\theta$ and the additive confound---\emph{exactly}, with no
small-gain approximation, when the regressor is the pre-window
output. If instead the available regressor is the full-period output
$\mathbf{R}_t=\mathbf{r}_{1,t}+\mathbf{r}_{2,t}
=(\mathbf{I}+\bM_t)(\mathbf{r}_{1,t}+\mathbf{v}_t)$, the two
projection ingredients follow from \eqref{eq:ondecomp} and the
orthogonality assumptions:
\begin{align}
\mathbb{E}[\mathbf{r}_{\mathrm{on}}\mathbf{R}^{\!\top}\,|\,\bgam_t]
&=\big[\mathbf{C}\boldsymbol{\Sigma}_1
-\theta\bM_t\boldsymbol{\Sigma}_1\nonumber\\
&\qquad-\theta(\mathbf{I}+\bM_t)
\boldsymbol{\Sigma}_v\big](\mathbf{I}+\bM_t)^{\!\top},
\label{eq:crossmom}\\
\mathbb{E}[\mathbf{R}\mathbf{R}^{\!\top}\,|\,\bgam_t]
&=(\mathbf{I}+\bM_t)(\boldsymbol{\Sigma}_1+\boldsymbol{\Sigma}_v)
(\mathbf{I}+\bM_t)^{\!\top},
\label{eq:regmom}
\end{align}
and dividing \eqref{eq:crossmom} by \eqref{eq:regmom} gives the
exact estimand
\begin{align}
\mathbf{B}(\bgam_t)=
\big[\mathbf{C}\boldsymbol{\Sigma}_1-\theta\,
\bM_t\boldsymbol{\Sigma}_1&-\theta(\mathbf{I}+\bM_t)
\boldsymbol{\Sigma}_v\big]\nonumber\\
&\times(\boldsymbol{\Sigma}_1+\boldsymbol{\Sigma}_v)^{-1}
(\mathbf{I}+\bM_t)^{-1},
\label{eq:estimand}
\end{align}
verified numerically in the replication code. Two consequences
matter. First, the confound no longer enters as a clean additive
intercept: even as $\boldsymbol{\Sigma}_v\to0$ the confound term is
$\mathbf{C}(\mathbf{I}+\bM_t)^{-1}$, itself gain-dependent, so the
constant-confound separation that drives
Theorem~\ref{thm:ident} is exact only for the pre-window regressor.
The theory below is therefore stated for the pre-window moment
\eqref{eq:momentclean}; when only full-period outputs are available,
\eqref{eq:estimand} defines the contaminated estimand. The
discrepancy from the clean case is $O(\|\bM\|)$ under bounded
$\mathbf{C}$---even at $\boldsymbol{\Sigma}_v=\mathbf{0}$ the
full-period estimand is
$(\mathbf{C}-\theta\bM_t)(\mathbf{I}+\bM_t)^{-1}$, still
gain-contaminated---while only the latent-noise component vanishes
with the clearing-window share $\boldsymbol{\Sigma}_v$. A small
clearing-window variance therefore does \emph{not} by itself license
the clean theorem; the convention analysis for the financial
application is in the companion paper \cite{woo2026arxiv}. The identification question is when the
observable map $\bgam\mapsto\mathbf{B}(\bgam)$ pins down
$(\bPhi,\mathbf{C})$.

\section{Identifiability}\label{sec:ident}

To first order in $\|\bL_t\|$ (the empirically relevant regime;
higher-order effects are treated as bias and quantified in
Section~\ref{sec:sims}), $\bM_t\approx\bL_t=\bPhi\,
\mathrm{diag}(\bgam_t)$ and \eqref{eq:momentclean} becomes entrywise
\begin{equation}
b_{ij}(t)= c_{ij}-\theta\,\phi_{ij}\,\gamma_{j,t}
+O(\|\bL\|^2),
\label{eq:linmoment}
\end{equation}
the observable slope varies \emph{linearly in the known signal}
$\gamma_{j,t}$, with the confound as intercept. Throughout,
$\tilde\gamma_{j,t}=(\gamma_{j,t}-\bar\gamma_j)/s_{\gamma_j}$ denotes
the standardized gain path (sample mean and standard deviation over
the estimation window), $\tilde{\boldsymbol{\Gamma}}\in\mathbb{R}^{T\times n}$ stacks them, and
$\boldsymbol{\Sigma}_r=\mathbb{E}[\mathbf{r}_{1,t}\mathbf{r}_{1,t}^{\!\top}]$.

\begin{theorem}[Identification and sharpness, linearized]
\label{thm:ident}
Parts (i)--(ii) concern identification from the conditional
coefficient map $t\mapsto b_{ij}(t)$ in \eqref{eq:linmoment}; part
(iii) concerns the pooled raw-data regression, whose rank condition
is $\lambda_{\min}(\mathbf{Q})>0$ and does not follow from (i)
alone.\\
(i) The pair $(c_{ij},\theta\phi_{ij})$ is identified if and only if
the design $[\mathbf{1},\ \gamma_{j,1:T}]$ has full column rank,
i.e., $\mathrm{Var}_t(\gamma_{j,t})>0$. (ii) If
$\mathrm{Var}_t(\gamma_{j,t})=0$, the set of observationally
equivalent pairs is the line
$\{(c_{ij}+\theta s\gamma_j,\ \phi_{ij}+s):s\in\mathbb{R}\}$; no
estimator can distinguish coupling from confound. (iii) In the joint regression across senders, the information matrix
of the interaction block is
$\mathbf{Q}=T^{-1}\mathbb{E}[\mathbf{Z}^{\!\top}
\mathbf{M}_{X}\mathbf{Z}]$, where the columns of $\mathbf{Z}$ are
$z_{j,t}=r_{1,j,t}\tilde\gamma_{j,t}$ and $\mathbf{M}_X$
residualizes the intercept and main effects; identification strength
is $\lambda_{\min}(\mathbf{Q})$. Under gains independent of the
pre-window signals,
$\mathbb{E}[z_{j}z_{k}]=[\boldsymbol{\Sigma}_r\circ\mathbf{G}]_{jk}$
(Hadamard product), with
$\mathbf{G}=T^{-1}\tilde{\boldsymbol{\Gamma}}^{\!\top}
\tilde{\boldsymbol{\Gamma}}$ the Gram matrix of standardized gain
paths: the conditioning of the gain design matters exactly to the
extent that the signals themselves are cross-correlated, and for
strongly factor-driven signals ($\boldsymbol{\Sigma}_r$ close to
rank one) $\kappa(\mathbf{G})$ becomes the governing factor of the
variance inflation.
\end{theorem}

Part (iii) is not a technicality: in the motivating application,
product launches make all gains ramp together, and
Section~\ref{sec:sims} shows the false-alarm rate of a signed
detector nearly doubling between a correlated staircase design
(condition number ${\sim}247$) and independent regime switching
(${\sim}12$) at equal per-channel variances.

\begin{IEEEproof}
(i) Stacking \eqref{eq:linmoment} over $t$ gives, per entry $(i,j)$,
the linear system $\mathbf{b}_{ij}=\mathbf{X}_j\boldsymbol{\beta}$
with $\mathbf{X}_j=[\mathbf{1},\ \gamma_{j,1:T}]\in\mathbb{R}^{T\times2}$
and $\boldsymbol{\beta}=(c_{ij},-\theta\phi_{ij})^{\!\top}$. The map
$\boldsymbol{\beta}\mapsto\mathbf{X}_j\boldsymbol{\beta}$ is
injective iff $\mathrm{rank}(\mathbf{X}_j)=2$, which holds iff
$\gamma_{j,t}$ is not constant, i.e.,
$\mathrm{Var}_t(\gamma_{j,t})>0$.
(ii) If $\gamma_{j,t}\equiv\gamma_j$, then for any $s\in\mathbb{R}$
the pair $(c_{ij}+\theta s\gamma_j,\ \phi_{ij}+s)$ generates the
identical moment:
\begin{equation}
(c_{ij}+\theta s\gamma_j)-\theta(\phi_{ij}+s)\gamma_j
=c_{ij}-\theta\phi_{ij}\gamma_j,
\label{eq:equivline}
\end{equation}
so the observationally equivalent set is a line and no statistic of
the data distinguishes its points.
(iii) Here $\mathbf{X}$ contains the intercept and the $n$ signal
main effects $r_{1,k,t}$, gains are centered
($\mathbb{E}[\tilde\gamma_j]=0$) and independent of signals, so
$\mathbb{E}[z_{j,t}r_{1,k,t}]=\mathbb{E}[r_{1,j}r_{1,k}]\,
\mathbb{E}[\tilde\gamma_j]=0$ and residualization leaves the
interaction block orthogonal to $\mathbf{X}$ in population. Then,
\begin{equation}
\mathbb{E}[z_{j,t}z_{k,t}]
=\mathbb{E}[r_{1,j}r_{1,k}]\,
\mathbb{E}[\tilde\gamma_{j}\tilde\gamma_{k}]
=[\boldsymbol{\Sigma}_r]_{jk}\,[\mathbf{G}]_{jk},
\label{eq:hadamard}
\end{equation}
i.e.\ the interaction Gram matrix is the Hadamard product
$\boldsymbol{\Sigma}_r\circ\mathbf{G}$; residualizing intercept and
main effects yields $\mathbf{Q}$ as stated. By the Schur product
theorem $\boldsymbol{\Sigma}_r\circ\mathbf{G}\succeq0$, and for the equicorrelation model
$\boldsymbol{\Sigma}_r=\sigma^2[(1-\varrho)\mathbf{I}
+\varrho\mathbf{1}\mathbf{1}^{\!\top}]$ (with $\varrho$ the signal
correlation, distinct from the spectral radius $\rho$),
\begin{align}
\boldsymbol{\Sigma}_r\circ\mathbf{G}
&=\sigma^2\big\{(1-\varrho)\,\mathrm{diag}(G_{jj})
+\varrho\,\mathbf{G}\big\},\nonumber\\
\lambda_{\min}(\boldsymbol{\Sigma}_r\circ\mathbf{G})
&\ge\sigma^2\big[(1-\varrho)\min_j G_{jj}
+\varrho\,\lambda_{\min}(\mathbf{G})\big]:
\label{eq:equicorr}
\end{align} as
$\varrho\uparrow1$ (factor-dominated signals) the bound is governed
by $\lambda_{\min}(\mathbf{G})$, while for $\varrho=0$
(uncorrelated signals) $\boldsymbol{\Sigma}_r\circ\mathbf{G}
=\sigma^2\mathrm{diag}(G_{jj})$ and gain-path correlation is
harmless. Degenerate $\mathbf{G}$ with factor-driven signals
returns linear combinations of senders to the unidentified case of
(ii).
\end{IEEEproof}

\begin{theorem}[Regular local identification, exact fixed point]
\label{thm:nonlinear}
Let $\mathbf{B}(\bgam)=\mathbf{C}-\theta\bM(\bgam)$ with
$\bM(\bgam)=(\mathbf{I}-\bPhi\,\mathrm{diag}(\bgam))^{-1}-\mathbf{I}$
be observed at gain regimes $\bgam^{(1)},\dots,\bgam^{(q)}$,
$\rho(\bPhi\,\mathrm{diag}(\bgam^{(k)}))<1$ for all $k$ (regimes
recur with positive probability, or are estimable by local
smoothing, so that the regime-specific moments are available). Define, for
each regime, the $n^2\times n^2$ Jacobian block
\begin{equation}
\mathbf{D}_k \equiv
\frac{\partial\,\mathrm{vec}\,\bM(\bgam^{(k)})}
{\partial\,\mathrm{vec}\,\bPhi}
=\big[\mathrm{diag}(\bgam^{(k)})\,\boldsymbol{\mathcal{R}}_k\big]^{\!\top}
\!\otimes\boldsymbol{\mathcal{R}}_k,
\label{eq:jacobian}
\end{equation}
where $\boldsymbol{\mathcal{R}}_k=(\mathbf{I}-\bPhi\,\mathrm{diag}
(\bgam^{(k)}))^{-1}$ (calligraphic to distinguish the resolvent from
the full-period output $\mathbf{R}_t$). Then full column rank of the stacked differenced Jacobian,
\begin{equation}
\mathrm{rank}\!\begin{bmatrix}
\mathbf{D}_2-\mathbf{D}_1\\ \vdots\\ \mathbf{D}_q-\mathbf{D}_1
\end{bmatrix}=n^2,
\label{eq:rankcond}
\end{equation}
is sufficient for local identification of $(\mathbf{C},\bPhi)$ at
the true parameters, and necessary for \emph{regular} (first-order)
local identification, i.e.\ on the set where the rank is locally
constant; rank deficiency at isolated points does not by itself
imply non-identification.
In particular $q\ge2$ distinct regimes are necessary; if all
$\bgam^{(k)}$ coincide, every $\mathbf{D}_k$ is identical, the
stacked matrix is zero, and constant gains again yield a continuum
of observationally equivalent pairs, whose first-order tangent
contains the lines of Theorem~\ref{thm:ident}(ii).
\end{theorem}

\begin{IEEEproof}
Differentiating
$\bM=(\mathbf{I}-\bL)^{-1}-\mathbf{I}$ along $d\bPhi$ with
$\bgam$ fixed at $\bgam^{(k)}$,
\begin{align}
d\bM&=\boldsymbol{\mathcal{R}}_k\,(d\bPhi)\,
\mathrm{diag}(\bgam^{(k)})\,\boldsymbol{\mathcal{R}}_k,
\label{eq:dM}\\
\mathrm{vec}(d\bM)&=\big[(\mathrm{diag}(\bgam^{(k)})
\boldsymbol{\mathcal{R}}_k)^{\!\top}\!\otimes
\boldsymbol{\mathcal{R}}_k\big]\mathrm{vec}(d\bPhi),
\label{eq:vecdM}
\end{align}
by $\mathrm{vec}(\mathbf{A}\mathbf{X}\mathbf{B})
=(\mathbf{B}^{\!\top}\!\otimes\mathbf{A})\,\mathrm{vec}\,
\mathbf{X}$, which gives \eqref{eq:jacobian}. The stacked observation map
$(\mathbf{C},\bPhi)\mapsto
\{\mathbf{B}(\bgam^{(k)})\}_{k=1}^q$ has Jacobian
$[\mathbf{I}_{n^2},-\theta\mathbf{D}_k]_{k=1,\dots,q}$ (row blocks).
Eliminating $\mathbf{C}$ by subtracting the first block row leaves
$-\theta[\mathbf{D}_k-\mathbf{D}_1]_{k\ge2}$ acting on
$\mathrm{vec}\,\bPhi$ alone; the implicit function theorem gives
local injectivity under \eqref{eq:rankcond} (and its necessity at
rank-regular points), and $\mathbf{C}$ is then
recovered from the first block. Necessity of $q\ge2$ and failure
under coinciding regimes are immediate. To first order
$\boldsymbol{\mathcal{R}}_k\approx\mathbf{I}$, so
$\mathbf{D}_k-\mathbf{D}_1\approx
\mathrm{diag}(\bgam^{(k)}-\bgam^{(1)})^{\!\top}\!\otimes\mathbf{I}$,
whose rank is $n\cdot\#\{j:\gamma_j^{(k)}\ne\gamma_j^{(1)}
\text{ for some }k\}$---recovering the per-column variance condition
of Theorem~\ref{thm:ident}(i) as the linearized special case.
\end{IEEEproof}

\begin{remark}[Scale, linearized implementation]
Only the products $\theta\phi_{ij}$ enter \eqref{eq:linmoment}; the
level of $\bPhi$ is identified once $\theta$ is calibrated or
normalized (e.g., own-channel reversal normalization). In the exact resolvent, powers of $\bPhi$ enter with different
orders, so unknown $\theta$ is a global scale ambiguity only to
first order; all spectral statements inherit the calibration of
$\theta$.
\end{remark}

\section{The Estimand Is a Resolvent Sensitivity}\label{sec:estimand}

Beyond first order, the object the interaction estimator
approximates is not $\phi_{ij}$ but (a gain-distribution-weighted
projection of) the sensitivity of the resolvent to the $j$th gain;
with discrete regime variation the coefficient is a secant rather
than a derivative, a distinction that matters for the second-stage
inversion below. Differentiating
$\bM=(\mathbf{I}-\bL)^{-1}-\mathbf{I}$ with
$d\bL=\bPhi\,\mathbf{e}_j\mathbf{e}_j^{\!\top}d\gamma_j$ and using
$d(\mathbf{I}-\bL)^{-1}=(\mathbf{I}-\bL)^{-1}(d\bL)
(\mathbf{I}-\bL)^{-1}$,
\begin{equation}
\frac{\partial \bM}{\partial\gamma_j}
=(\mathbf{I}-\bL)^{-1}\,\bPhi\,\mathbf{e}_j\mathbf{e}_j^{\!\top}\,
(\mathbf{I}-\bL)^{-1},
\label{eq:resolvent}
\end{equation}
which is nonzero for every pair $(i,j)$ connected by a directed path
through the network, not only for direct edges: expanding both
resolvents by \eqref{eq:neumann}, the $(i,j)$ entry of
\eqref{eq:resolvent} collects all walks $j\to\cdots\to i$ that pass
through the perturbed gain. At first order in $\|\bL\|$,
$\partial\bM/\partial\gamma_j\to\bPhi\,
\mathbf{e}_j\mathbf{e}_j^{\!\top}$, recovering
\eqref{eq:linmoment}. This is a feature for the monitoring
problem---transmission alarms and spectral margins are resolvent
objects---and a caveat for topology recovery: a signed detector
applied entrywise flags 2-hop-reachable zero entries at well above
nominal size while holding approximately nominal size on unreachable
entries (Section~\ref{sec:sims}). Direct-edge recovery is therefore
a second-stage inverse problem---inverting \eqref{eq:resolvent} over
the detected reachability pattern---which we formulate but do not
pursue empirically in this paper.

\subsection{Which gains identify which entries}
\label{sec:whichgains}

Theorem~\ref{thm:nonlinear} is an all-or-nothing rank condition. The
following refines it to a per-entry, graph-readable map from gain
variation to identified quantities, and is exact rather than
first-order. Write $\mathbf{m}_{\cdot l}$ for the $l$th column of
$\bM$ and $\mathbf{r}_{l\cdot}$ for the $l$th row of
$\mathbf{I}+\bM$, evaluated at the operating gain $\bgam$.

\begin{theorem}[Cross identification from varying gains]
\label{thm:cross}
Let $\mathcal{V}\subseteq\{1,\dots,n\}$ index the gains that vary
near the operating point $\bgam$, and assume:
(a) $\gamma_l>0$ for every
$l\in\mathcal{V}$ (at $\gamma_l=0$ the undivided sensitivity
$\partial\bM/\partial\gamma_l=(\mathbf{I}-\bL)^{-1}\bPhi\mathbf{e}_l
\mathbf{e}_l^{\!\top}(\mathbf{I}-\bL)^{-1}$ applies instead); and
(b) a \emph{local-support condition}: the support of the gain
process contains a relatively open neighborhood of $\bgam$ in the
coordinates $\bgam_{\mathcal{V}}$ (equivalently, local gain
increments span $\mathbb{R}^{|\mathcal{V}|}$), so that every
coordinate partial derivative $\partial\bM/\partial\gamma_l$,
$l\in\mathcal{V}$, is recoverable from the observable coefficient
map. If the gains move only along a lower-dimensional path---the
empirically relevant case of strongly correlated gain
paths---only the directional derivatives along that path are
observed, and the conclusions below hold for the coordinate
derivatives only once (b) is met. Then:
(i) the exact sensitivity factorizes as the rank-one outer product
\begin{equation}
\frac{\partial\bM}{\partial\gamma_l}
=\frac{1}{\gamma_l}\,\mathbf{m}_{\cdot l}\,\mathbf{r}_{l\cdot},
\qquad l\in\mathcal{V};
\label{eq:rankone}
\end{equation}
(ii) provided the nondegeneracy condition
$\mathbf{m}_{\cdot l}\ne\mathbf{0}$ holds (node $l$ has at least one
outgoing feedback path; otherwise the sensitivity vanishes
identically and only that fact is learned), the $(l,l)$ entry pins
the factorization scale: with
$p_{ll}=[\partial\bM/\partial\gamma_l]_{ll}$, the diagonal entry
solves $x(1+x)=\gamma_l\,p_{ll}$, whose nonnegative root is
$M_{ll}$---and when $M_{ll}=0$ the scale is pinned instead by the
row normalization $r_{ll}=1$: the $l$th column of the derivative
slab equals $\mathbf{m}_{\cdot l}/\gamma_l$ directly, and
$\mathbf{r}_{l\cdot}$ follows from any nonzero entry of
$\mathbf{m}_{\cdot l}$; hence \emph{the full column $\mathbf{m}_{\cdot l}$ and the
full row $\mathbf{r}_{l\cdot}$ are exactly identified for every
$l\in\mathcal{V}$ with an outgoing path}, so that a
\emph{guaranteed point-identified subset} of $\bM$ is the cross
pattern $\{(i,j): i\in\mathcal{V}'\text{ or }j\in\mathcal{V}'\}$,
$\mathcal{V}'=\{l\in\mathcal{V}:\mathbf{m}_{\cdot l}\ne\mathbf{0}\}$;
(iii) in the linearized regime the statement reduces to
Theorem~\ref{thm:ident}: $\partial\bM/\partial\gamma_l\to
\bPhi\mathbf{e}_l\mathbf{e}_l^{\!\top}$, only column $l$ of $\bPhi$
is exposed, and entries with $i,j\notin\mathcal{V}$ are unidentified
at first order. Whether the complement block is identified in the
exact model through higher-order joint variation is not settled by
this theorem: (i)--(ii) exhibit a point-identified subset, not the
full identified set.
\end{theorem}

\begin{IEEEproof}
(i) is \eqref{eq:resolvent} rewritten:
$(\mathbf{I}-\bL)^{-1}\bPhi\mathbf{e}_l
=(\mathbf{I}-\bL)^{-1}\bL\mathbf{e}_l/\gamma_l
=\mathbf{m}_{\cdot l}/\gamma_l$ and
$\mathbf{e}_l^{\!\top}(\mathbf{I}-\bL)^{-1}=\mathbf{r}_{l\cdot}$.
(ii) The $(l,l)$ entry of \eqref{eq:rankone} is
$M_{ll}(1+M_{ll})/\gamma_l$; the map $x\mapsto x(1+x)$ is strictly
increasing on $x\ge0$, so the root is unique under the maintained
nonnegativity, and dividing the observed rank-one factors by the
resulting scale recovers column and row exactly. (iii) follows by
$\bM\to\bL$, $(\mathbf{I}+\bM)\to\mathbf{I}$. All three statements
are verified numerically in the replication code.
\end{IEEEproof}

The graph reading: \emph{a varying gain at a node with outgoing
feedback illuminates every transmission path into it and out of
it}---row and column $l$ of the resolvent---and pins its own
diagonal; a varying gain at a sink node (no outgoing path) reveals
only that fact, and gains that never move leave the
$\mathcal{V}^c\times\mathcal{V}^c$ block of $\bM$ dark at first
order. Two consequences follow. For \emph{monitoring},
partial gain variation is enough for partial surveillance under the
local-support condition (b): all
transmission entries $M_{ij}$ with a varying sender or receiver are
point-identified without any support assumption on $\bPhi$, which is exactly
the set of alarms a supervisor of the varying complexes needs. For
\emph{estimation}, \eqref{eq:rankone} implies that the derivative
slab associated with gain $l$ is rank one in population---so a
rank-one (SVD) projection of the estimated full interaction slab
(all $r_{1,k}\tilde\gamma_l$ regressors) is a noise-reduced
estimator of $(\mathbf{m}_{\cdot l},\mathbf{r}_{l\cdot})$ that
bypasses the fixed-point inversion \eqref{eq:twostage}
entirely. By Eckart--Young the rank-one \emph{projection
subproblem} is globally optimal---so that step, unlike the nonlinear
least-squares fit, has no local-minimum pathology; no such claim is
made for the full estimator after scale recovery, assembly, and
fallback. We implement a columns-only version (rows are redundant
when all columns are available and are usable as a diagnostic) and
find that, among root-success paths in the design of
Section~\ref{sec:sims}, it attains lower spectral bias and far
lower dispersion than the selected two-stage inversion; numbers,
Monte Carlo accounting, and failure handling are reported in
Section~\ref{sec:uq}. The implementation uses global standardized
slopes, so the secant caveat of Section~\ref{sec:algo} applies to
its estimand; its distribution theory alongside
Theorem~\ref{thm:bootstrap} is a natural next step. Whether the dark block is identified through
higher-order joint variation of $\bgam_{\mathcal{V}}$ in the exact
model is an open question we state precisely: all
$\bgam_{\mathcal{V}}$-derivatives of $\mathbf{B}$ expose products of
cross-pattern entries only, so any identification of the dark block
must come through the nonlinear dependence of those cross entries on
the full $\bPhi$.

\section{Estimation, Algorithms, and Monitoring}\label{sec:algo}

\subsection{Interaction estimator}

Algorithm~\ref{alg:ident} implements the moment
\eqref{eq:linmoment}. For receiver $i$, the regression
\begin{equation}
r_{\mathrm{on},i,t}=\alpha_i+\sum_{j=1}^{n}
\big(b_{ij}\,r_{1,j,t}+d_{ij}\,r_{1,j,t}\,\tilde\gamma_{j,t}\big)
+\varepsilon_{i,t}
\label{eq:regression}
\end{equation}
separates the constant confound (absorbed by $b_{ij}$ together with
the gain-mean response) from the gain-tracking response
$d_{ij}=-\theta\,\phi_{ij}\,s_{\gamma_j}+O(\|\bL\|^2)$, where
$s_{\gamma_j}$ is the standardization scale of
$\tilde\gamma_{j,t}=(\gamma_{j,t}-\bar\gamma_j)/s_{\gamma_j}$.
Heteroskedasticity-robust variances give entrywise $t$-statistics;
the signed decision rule $t_{ij}<-z_{1-\alpha}$ exploits the known
direction of the reversal (coupling makes the interaction
\emph{negative}), which a confound of either sign cannot mimic
without tracking $\bgam_t$ (Assumption~\ref{as:exo}).

Two scope qualifications attach to \eqref{eq:regression}. First, it
is the \emph{low-gain} implementation of the identification theory:
it matches Theorem~\ref{thm:ident} exactly and
Theorem~\ref{thm:nonlinear} to first order; an exact-fixed-point
nonlinear least-squares estimator fitting
$\mathbf{B}(\bgam)=\mathbf{C}-\theta[(\mathbf{I}-\bPhi\,
\mathrm{diag}(\bgam))^{-1}-\mathbf{I}]$ across regimes is the
natural second stage, with the linear estimates as initialization.
Second, \eqref{eq:regression} includes the same-index interactions
$r_{1,j,t}\tilde\gamma_{j,t}$ only, while beyond first order the
expansion of $(\mathbf{C}-\theta\bM(\bgam))\mathbf{r}_1$ generates
all $n^2$ cross terms $r_{1,k,t}\tilde\gamma_{j,t}$ with
coefficients $-\theta\,\partial M_{ik}/\partial\gamma_j$; with
correlated regressors the same-index OLS coefficient is then a
projection mixture, so Algorithm~\ref{alg:ident} estimates the
first-order direct-coupling approximation with an
$O(\|\bL\|^2)$ projection bias, not the exact sensitivity. The full
$n^2$-interaction regression is feasible at moderate $n$
(Section~\ref{sec:sims} reports it: it reduces the reachable-zero
flag rate from $0.39$ to $0.27$ at unchanged size on unreachable
zeros) and is the cleaner implementation when $T$ permits.

\begin{algorithm}[t]
\caption{Output-only network identification from known gains}
\label{alg:ident}
\begin{algorithmic}[1]
\REQUIRE outputs $\{\mathbf{r}_{1,t},\mathbf{r}_{\mathrm{on},t}\}_{t=1}^T$,
gains $\{\bgam_t\}$, reversal share $\theta$, size $\alpha$
\STATE \textbf{if} any $s_{\gamma_j}$ is zero or negligible
\textbf{then} declare column $j$ unidentified
(Theorem~\ref{thm:ident}(ii)) and drop it \textbf{end if}
\STATE standardize remaining gain paths $\tilde\gamma_{j,t}\gets
(\gamma_{j,t}-\bar\gamma_j)/s_{\gamma_j}$; form
$\mathbf{G}=T^{-1}\tilde{\boldsymbol{\Gamma}}^{\!\top}
\tilde{\boldsymbol{\Gamma}}$
\STATE report $\lambda_{\min}(\hat{\mathbf{Q}})$ of the residualized
interaction Gram matrix as the identification-strength diagnostic
(abstain if below a preset floor), with
$\mathrm{cond}(\mathbf{G})$ as the secondary design diagnostic
(Theorem~\ref{thm:ident}(iii))
\FOR{$i=1,\dots,n$}
\STATE OLS of \eqref{eq:regression}; collect $d_{ij}$ with HC
errors for independent designs, HAC or moving-block-bootstrap
errors for dependent data; form $t_{ij}$
\ENDFOR
\STATE $\hat\phi^{(1)}_{ij}\gets -d_{ij}/(\theta\, s_{\gamma_j})$
(first-order/low-gain estimate; exact recovery is the second-stage
inversion of Sec.~\ref{sec:uq}); flag couplings with
$t_{ij}<-z_{1-\alpha}$ (marginal, exploratory tests---see
multiplicity note below)
\STATE assemble $\hat\bL_t=\hat\bPhi^{(1)}\mathrm{diag}(\bgam_t)$
with nonnegative projection $\hat\bL_t\gets\max(\hat\bL_t,0)$;
$\hat\bM_t=(\mathbf{I}-\hat\bL_t)^{-1}-\mathbf{I}$
\RETURN $\hat\bPhi^{(1)}$, flags, screening traces
$\{\rho(\hat\bL_t)\}$, $\{\hat M_{ij,t}\}$
\end{algorithmic}
\end{algorithm}

\subsection{Estimation theory for the first-order estimator}
\label{sec:esttheory}

Theorems~\ref{thm:ident}--\ref{thm:nonlinear} are population
statements; the following closes the gap to
Algorithm~\ref{alg:ident}. Treat the gain path as a deterministic
design (Assumption~\ref{as:exo}) with nondegenerate limits: the
standardized design satisfies
$T^{-1}\tilde{\mathbf{Z}}^{\!\top}\mathbf{M}_X\tilde{\mathbf{Z}}
\to\mathbf{Q}\succ0$ with $\tilde{\mathbf{Z}}$ the interaction block
and $\mathbf{M}_X$ the main-effect residualizer.

\begin{theorem}[Consistency and asymptotic normality]
\label{thm:estimation}
Let Assumptions~\ref{as:exo}--\ref{as:stoch} hold, strengthened so
that $\{(\mathbf{r}_{1,t},\mathbf{v}_t,\mathbf{w}_t)\}$ is
$\alpha$-mixing with coefficients satisfying
$\sum_k\alpha(k)^{\delta/(2+\delta)}<\infty$ and
$\mathbb{E}\|\cdot\|^{4+2\delta}<\infty$ for some $\delta>0$. Assume further a bounded deterministic design whose lagged
design-product averages converge (so that long-run limits exist for
the triangular score array), a uniform stability margin
$\sup_t\rho(\bL_t)\le\bar\rho<1$. The CLT below is stated for the
score array centered at its \emph{date-specific means},
$\mathbf{s}_{t,T}-\mathbb{E}[\mathbf{s}_{t,T}]$---a
triangular-array formulation, because under the misspecified pooled
projection the omitted cross-interactions generally give the scores
nonzero, date-dependent deterministic means---with the pseudo-true
coefficient defined by the \emph{aggregate} normal equations,
$T^{-1}\sum_t\mathbb{E}[\mathbf{s}_{t,T}]\to\mathbf{0}$; the HAC
estimator is applied to the demeaned scores, and its bandwidth
conditions are assumed to hold for the centered array. Let
$\hat{\mathbf{d}}_i$ collect the interaction coefficients for
receiver $i$ and $\mathbf{d}^{*}_i$ their pseudo-true values, which
satisfy $d^{*}_{ij}=-\theta\,\phi_{ij}\,s_{\gamma_j}
+O(\|\bL\|^2)$. Then, by the Frisch--Waugh--Lovell reduction to the
residualized interaction block,
\begin{equation}
\sqrt{T}\,(\hat{\mathbf{d}}_i-\mathbf{d}^{*}_i)
\;\xrightarrow{d}\;
\mathcal{N}\!\big(\mathbf{0},\,
\mathbf{Q}^{-1}\boldsymbol{\Omega}_i\mathbf{Q}^{-1}\big),
\label{eq:clt}
\end{equation}
where $\mathbf{Q}$ is the residualized interaction information
matrix defined above and $\boldsymbol{\Omega}_i$ is the long-run
covariance of the residualized interaction score
$\tilde{\mathbf{z}}_t\varepsilon_{i,t}$, consistently estimable by
HAC smoothing with standard bandwidth conditions \cite{newey1987}
or a design-preserving block bootstrap (see
Theorem~\ref{thm:bootstrap}).
Consistency is for the pseudo-true parameter; its distance to the
structural first-order target is the projection bias of
Section~\ref{sec:estimand}, of order
$\sup_t\|\bPhi\,\mathrm{diag}(\bgam_t)\|^{2}$ with a constant that
grows as $\lambda_{\min}(\mathbf{Q})\to0$---a modeling error, not a
statistical one, and it does not shrink with $T$.
\end{theorem}

\begin{IEEEproof}[Proof sketch]
Multiplying stationary mixing signals by a bounded deterministic
gain design yields a triangular \emph{mixing array} of scores, not
a stationary sequence; the assumed lag-product design limits give
existence of the long-run covariance, centering at the
date-specific means makes the array mean zero by construction
(while the aggregate normal equations pin the pseudo-true
coefficient), and the CLT for mixing arrays
\cite[Thm.~5.20]{white2001} applies to the centered residualized
interaction score. Frisch--Waugh--Lovell reduces the full regression to the
interaction block with information matrix $\mathbf{Q}$, yielding
\eqref{eq:clt} with the sandwich form; HAC consistency under the
same mixing, moment, and bandwidth conditions is \cite{newey1987}. The pseudo-true gap
statement is the entrywise expansion \eqref{eq:linmoment} evaluated
under the projection, whose remainder collects the omitted
cross-interaction terms of Section~\ref{sec:algo}.
\end{IEEEproof}

Two design corollaries matter in practice, and both are verified
numerically. First, under a \emph{fixed} gain design the
$\sqrt{T}$ rate is visible end to end: tripling $T$ over the sequence $T\in\{250,750,2250,6750\}$ ($120$
replications each) shrinks the RMSE to the pseudo-true parameter by
factors $1.75/1.74/1.68$ against the theoretical
$\sqrt{3}\approx1.73$. Second, under \emph{random} gain designs the
pseudo-true parameter is itself design-dependent, so unconditional
convergence stalls at the design variability floor---the practical
reading is that inference should condition on the realized gain
path, exactly as the fixed-design frame of
Theorem~\ref{thm:estimation} does.

\subsection{Monitoring statistics and dynamic threshold}

Three alarms follow from the assembled matrices, one per risk object.
\emph{Cycle margin}: $\rho(\hat\bL_t)$ against the static boundary
$1$ and the dynamic boundary below. \emph{Transmission}: off-diagonal
resolvent entries $\hat M_{ij,t}$, the displacement imported by
channel $i$ per unit of channel $j$'s innovation. \emph{Persistence}:
chaining periods through the partial correction, the uncorrected
displacement $\mathbf{e}_t$ evolves as
\begin{equation}
\mathbf{e}_{t+1}
=\big[(1-\theta)\mathbf{I}-\theta\bM\big]\mathbf{e}_t
+\bM\,\boldsymbol{\varepsilon}_{t+1},
\label{eq:chain}
\end{equation}
whose derivation is: a share $\theta$ of the carried displacement is
corrected (leaving $(1-\theta)\mathbf{e}_t$), and the new clearing
window adds $\bM$ applied to the period innovation net of the
correction flow, $\boldsymbol{\varepsilon}_{t+1}-\theta\mathbf{e}_t$.
For $\lambda\in\sigma(\bL)$ the matching chain eigenvalue is
\begin{equation}
1-\theta-\theta\,\frac{\lambda}{1-\lambda}
=1-\frac{\theta}{1-\lambda},
\label{eq:chaineig}
\end{equation}
so with a real spectrum and $\rho(\bL)<1$ the chain is covariance
stationary iff $|1-\theta/(1-\lambda)|<1$ for all $\lambda$; for
$\lambda\in[0,1)$ this is $\lambda<1-\theta/2$, negative eigenvalues
satisfy it automatically for $\theta\in(0,1]$, and the Perron root
binds. The operative stability boundary for monitoring is therefore
$\rho(\bL)<1-\theta/2$, strictly tighter than the within-period
boundary---at $\theta=0.88$ (the case-study calibration), $0.56$
versus $1$. Three restrictions delimit this statement: it freezes
$\bL$ over the horizon (for genuinely time-varying gains, stability
is governed by products of transition matrices, for which pointwise
spectral conditions are neither necessary nor sufficient in
general); it assumes a common scalar correction $\theta$; and the
reduction to $1-\theta/2$ uses a real spectrum. It is a monitoring
heuristic for slowly varying $\bL_t$, not a general stability
theorem.

\begin{remark}[Network-analytic reading]
The three monitoring objects have close counterparts in network
analysis. The resolvent $(\mathbf{I}-\bL)^{-1}=\sum_k\bL^k$ is the
walk-generating kernel of Katz's status index \cite{katz1953}, so
the transmission entry $M_{ij}$ is a Katz-type walk influence of
channel $j$ on channel $i$ (a communicability entry; Katz's index
proper aggregates these entries per node). The left and
right Perron vectors $\mathbf{u},\mathbf{v}$ that already enter the
delta method \eqref{eq:deltamethod} are the corresponding
influence and susceptibility centralities, and for a simple Perron
root with $\gamma_j>0$ the correspondence is exact:
$\partial\rho(\bPhi\,\mathrm{diag}\,\bgam)/\partial\gamma_j
=\rho\,u_jv_j/(\gamma_j\,\mathbf{u}^{\!\top}\mathbf{v})$ under the
normalization $\mathbf{u}^{\!\top}\mathbf{v}>0$, so the
product $u_jv_j$ ranks channels by the elasticity of the
spectral margin to their gains, giving the
monitor a principled answer to \emph{which} node drives an
approaching alarm. And the dynamic boundary
$\rho(\bL)<1-\theta/2$ is a limited algebraic analogy to an
epidemic
threshold in network spread models---spectral radius of the contact
structure against a recovery rate, here the overnight correction
share $\theta$---valid only under the frozen-$\bL$, common
scalar-$\theta$, real-spectrum assumptions of the dynamic
threshold, not an epidemic-threshold identity. Beyond
interpretation, this correspondence points
at the natural regularization class for moderate and large $n$:
network-analytic structure on $\bPhi$ (low-rank latent positions,
block/community structure) reduces the parameter count from $n^2$
and correspondingly weakens the gain-richness demanded by the rank
condition of Theorem~\ref{thm:nonlinear}---a quantitative
structure-for-data exchange we leave to future work.
\end{remark} Algorithm~\ref{alg:monitor} packages the three alarms
with a consecutive-exceedance rule.

\begin{algorithm}[t]
\caption{Rolling spectral and transmission monitoring}
\label{alg:monitor}
\begin{algorithmic}[1]
\REQUIRE data $\{\mathbf{r}_{1,t},\mathbf{r}_{\mathrm{on},t},
\bgam_t\}$, calibration $\theta$, window $W$, thresholds $\bar\rho$
(default $1-\theta/2$; heuristic when its frozen-$\bL$/scalar-$\theta$/
real-spectrum conditions fail), $\bar m$, exceedance count
$K_{\mathrm{run}}$, size $\alpha$, bootstrap block length $b$
\FOR{each period $t\ge W$}
\STATE run Algorithm~\ref{alg:ident} on window $t-W+1,\dots,t$
\STATE \emph{(spectral branch:)} $\mathcal{A}_\rho(t)\gets
\mathbb{I}\{\rho(\hat\bL_t^{(2)})>\bar\rho_t^{*}\}$, where
$\hat\bL_t^{(2)}=\hat\bPhi^{(2)}\mathrm{diag}(\bgam_t)$ uses the
inversion \eqref{eq:twostage}---the estimator whose bootstrap
calibration Theorem~\ref{thm:bootstrap} covers; the experimental
SVD-slab point estimator of Sec.~\ref{sec:uq} may be reported
alongside but carries no inferential claim---on the support selected
by the \emph{model-selection threshold} $c_T$ of
Theorem~\ref{thm:bootstrap} (not the fixed-$\alpha$ screen of
line~7), and $\bar\rho_t^{*}$ is a design-preserving
block-bootstrap basic-interval critical value for
$H_0:\rho^{(2)*}\le\bar\rho$---inference on the pseudo-true
functional $\rho^{(2)*}$, whose distance to the structural $\rho$
is the secant/projection gap of Sec.~\ref{sec:uq} ($90\%$ empirical
coverage there; pointwise validity in Theorem~\ref{thm:bootstrap},
uniform validity open);\quad
$\mathcal{A}_m(t)\gets
\mathbb{I}\{\exists\,i\ne j:\hat M_{ij,t}>\bar m\ \text{and}\
 t_{ij}<-z_{1-\alpha}\}$
\STATE raise alarm if $\mathcal{A}_\rho$ or $\mathcal{A}_m$ persists
for $K_{\mathrm{run}}$ consecutive windows
\ENDFOR
\end{algorithmic}
\end{algorithm}

Entrywise screening over $n^2$ edges and overlapping windows
requires familywise or false-discovery control, and the
consecutive-exceedance count $K_{\mathrm{run}}$ trades average run length
against detection delay; calibrating both is part of the sequential
design, not an afterthought, and is included in the UQ program
below.

\subsection{Uncertainty quantification for spectral alarms}
\label{sec:uq}

The alarms of Algorithm~\ref{alg:monitor} are nonlinear functionals
of noisy regression coefficients, and a point estimate of
$\rho(\hat\bL)$ without an error bar is not an implementable
statistic. Three layers are involved.

\emph{(a) From coefficients to entries.} The interaction estimates
$\hat d_{ij}$ carry an HC covariance from \eqref{eq:regression};
entry estimates inherit it through the linear map of
Algorithm~\ref{alg:ident}.

\emph{(b) From entries to $\rho$.} For an irreducible nonnegative
matrix the Perron root is a simple real eigenvalue; with left and
right Perron vectors $\mathbf{u},\mathbf{v}$, first-order
eigenvalue perturbation gives
\begin{equation}
\frac{\partial\rho}{\partial \ell_{ij}}
=\frac{u_i\,v_j}{\mathbf{u}^{\!\top}\mathbf{v}},
\qquad
\mathrm{se}^2(\hat\rho)\approx
\sum_{i,j}\Big(\frac{u_i v_j}{\mathbf{u}^{\!\top}\mathbf{v}}\Big)^{\!2}
\mathrm{Var}(\hat\ell_{ij}),
\label{eq:deltamethod}
\end{equation}
with the obvious quadratic form under correlated entry errors.

\emph{(c) The plug-in trap, quantified.} Because the first-stage
estimand is the resolvent sensitivity \eqref{eq:resolvent}, entrywise
estimates assembled directly into a matrix are \emph{network-biased}:
indirect paths inflate entries, and nonnegative projection converts
zero-entry noise into positive mass that accumulates in the spectral
radius. In the design of Section~\ref{sec:sims} (true
$\rho=0.30$), the naive plug-in gives
$\mathbb{E}[\hat\rho]=0.73$; sparsifying by the signed $t$-rule does
not repair it ($0.72$), because the bias sits in the \emph{detected}
entries. As an \emph{ad hoc bias correction}---the global
interaction slope is a secant across gain levels rather than a
derivative at a specified operating point, so this inversion is
heuristic rather than an exact second stage---the fixed-point iteration
\begin{equation}
\hat\bPhi^{(2)}\leftarrow(\mathbf{I}-\hat\bL)\,\hat{\mathbf{S}}\,
\mathrm{diag}(\hat h)^{-1},\quad
\hat h_j=[(\mathbf{I}-\hat\bL)^{-1}]_{jj},
\label{eq:twostage}
\end{equation}
evaluated at the mean gain---formally, $\hat\bPhi^{(2)}$ solves the
estimating equation
$\bPhi=(\mathbf{I}-\bPhi\,\mathrm{diag}(\bar\bgam))\,
\hat{\mathbf{S}}\,\mathrm{diag}(\hat h(\bPhi))^{-1}$ on the selected
support, with the nonnegativity projection assumed inactive at the
solution (strict interiority)---where $\hat{\mathbf{S}}$ collects
the sensitivity estimates; it recovers entries well
(mean absolute error $0.042$ on $\bPhi$) and halves the spectral
bias ($\mathbb{E}[\hat\rho]=0.47$); restricting to the $t$-detected
support gives $0.42$, and even the oracle support retains $+0.07$
bias with delta-method coverage of only $61\%$ (anti-conservative),
because \eqref{eq:deltamethod} ignores the support-selection step
and the nonneg projection. The rank-one slab estimator of
Section~\ref{sec:whichgains} sharpens the point-estimation frontier.
Implemented as: rescale the estimated full-interaction slab to
$\hat{\mathbf{P}}_l=-\hat{\mathbf{D}}_l/(\theta\,s_{\gamma_l})$
(so that in population $\hat{\mathbf{P}}_l$ targets the secant
analogue of $\partial\bM/\partial\gamma_l$; the scale-pinning
identity is exact only for the local derivative, hence heuristic
here), project to rank one by SVD---the singular-vector sign
ambiguity is resolved by orienting the pair so that the pinned
diagonal entry of the projected slab is nonnegative---then pin the
scale by $x(1+x)=c_l$ with
$c_l=\gamma_l[\hat{\mathbf{P}}_l]_{ll}$ read from the
rank-one-projected slab: real roots require $c_l\ge-1/4$, and the
branch consistent with the maintained nonnegativity,
$x=\big({-1+\sqrt{1+4c_l}}\big)/2$, is admissible ($x\ge0$)
precisely when $c_l\ge0$. \emph{Root failure} is defined as
$c_l<0$: noise has pushed the projected diagonal outside the
admissible cone. (For a channel whose true $M_{ll}=0$ the scale is
pinned instead by the $r_{ll}=1$ normalization of
Theorem~\ref{thm:cross}(ii); the implementation does not attempt to
distinguish a near-zero diagonal from noise and routes such
columns through the same fallback.) Assemble the surviving columns
into $\hat\bM$, requiring every gain to vary with positive
operating level. In this experiment ($100$ Monte Carlo paths, an
exception to the $200$ used elsewhere), the scale root succeeded on
$84$ paths; \emph{conditional on root success}, spectral bias is
$+0.096$ with standard deviation $0.034$ and $\bPhi$ RMSE
$0.056$. Because root success may select easier samples, we report
the matched and unconditional accounting explicitly: on the
\emph{same} $84$ root-success paths the two-stage inversion has
spectral bias $+0.135$ (sd $0.176$) against the SVD slab's
$+0.096$ (sd $0.034$); over \emph{all} $100$ paths, the operational
estimator---SVD slab with two-stage fallback on the $16$
root-failure paths---has bias $+0.098$ (sd $0.058$) against
$+0.131$ (sd $0.169$) for the two-stage inversion alone. The SVD
slab is thus near the oracle-support benchmark ($+0.07$) without
oracle knowledge, and the advantage survives the fallback mixture.
We therefore label it the
\emph{preferred experimental point estimator in this design}:
its calibrated inference is \emph{not} established---the $90\%$
bootstrap result above concerns the two-stage pipeline only, and
extending Theorem~\ref{thm:bootstrap} to the SVD branch would
additionally require an isolated pseudo-true top singular value, a
scale root bounded away from the failure boundary, and an
asymptotically constant fallback branch, none of which we assume
here.
The exact-model nonlinear least-squares
estimator suggested by Theorem~\ref{thm:nonlinear} completes the
picture from both ends: initialized at the truth it is nearly
unbiased ($\bPhi$ RMSE $0.029$, spectral bias $+0.04$)---an
\emph{infeasible oracle diagnostic} showing the sample carries the
information the local-identification theorem promises, not a usable
benchmark---but initialized at the first-order estimates it lands on
spectral-inflating ridges of the nonconvex criterion ($\bPhi$ RMSE
$0.153$, spectral bias $+0.66$, worse than the two-stage inversion
it was meant to refine).

\emph{Optimization path and target.} The exact estimator minimizes
the regime-map least-squares criterion after profiling out the
linear nuisance $\mathbf C$: for fixed $\bPhi$,
$\mathbf C^*(\bPhi)$ is available in closed form
(variable projection \cite{golub1973}). Regular local
identification gives an isolated local target near the truth but
does not make the criterion globally convex or select its basin.
The implemented diagnostic path is therefore explicit:
(1) estimate the regime maps $\hat{\mathbf B}_k$; (2) construct
either the first-order or rank-one SVD-slab initializer; (3) profile
$\mathbf C$; and (4) refine $\bPhi$ with a box-constrained
trust-region reflective solver, retaining the unpenalized criterion
and spectral radius for every endpoint. With first-order
initialization, profiling does not repair the geometry: the profiled
and joint fits land in the same spectral-inflating branch (spectral
bias $+0.54$ for both, criterion roughly $34\times$ its value at
the truth), so the nonconvexity is intrinsic to $\bPhi$. With the
SVD initializer, the diagnostic reaches the low-criterion branch on
24 of 25 paths, as reported below. No data-driven rule is claimed
to identify a global minimizer or to detect the failed path.
Proposition~\ref{prop:mmgeom} supplies a distinct, theoretically
convergent alternative: exact proximal-DCA solves over a hard
spectral feasible set---not the TRF implementation used in the
diagnostic experiments. The Jacobian of Theorem~\ref{thm:nonlinear}
remains available for local Gauss--Newton steps.

\begin{proposition}[Log-coordinate geometry and convergent MM
refinement]\label{prop:mmgeom}
Impose an entrywise floor $\bPhi\ge\varepsilon\mathbf{1}
\mathbf{1}^{\!\top}$, $\varepsilon>0$ (dense approximation;
$O(\varepsilon)$ bias under the uniform stability margin), write
$x=\log\bPhi$ entrywise, and restrict to
$\mathcal{X}=[\log\varepsilon,\log\Phi_{\max}]^{n\times n}\cap
\{x:\rho(e^{x}\,\mathrm{diag}(\bgam_{\max}))\le\bar\rho\}$ with
$\bar\rho<1$ and $\bgam_{\max}>0$ entrywise; by monotonicity of the
Perron root in the entries, every gain regime then satisfies
$\rho(e^{x}\,\mathrm{diag}(\bgam^{(k)}))\le\bar\rho$ on
$\mathcal{X}$. A sparse fixed-support extension requires
irreducibility and a simple Perron root of the max-gain matrix,
which support selection alone does not imply. Then:
(i) every entry of $\bM(e^{x},\bgam^{(k)})$ is a nonnegative
walk-series---a locally uniform limit of posynomials, equivalently a
sum of exponentials of affine functions of $x$---hence convex in
$x$ on $\mathcal{X}$; it is also real-analytic there by the
resolvent representation and the uniform stability margin;
(ii) $x\mapsto\log\rho(e^{x}\,\mathrm{diag}(\bgam_{\max}))$ is
convex \cite{kingman1961}, so $\mathcal{X}$ is compact convex, and
the floor makes $e^{x}\,\mathrm{diag}(\bgam_{\max})$ positive,
hence the Perron root simple and the constraint smooth;
(iii) the profiled criterion
$J(x)=\tfrac12\sum_k\|\hat{\mathbf{A}}_k
+\theta(\bM_k(x)-\bar{\bM}(x))\|_F^2$,
with $\hat{\mathbf{A}}_k$ the regime maps centered across regimes and
$\bar{\bM}$ the across-regime mean, admits the explicit
difference-of-convex decomposition $J=S-C$ with convex
\begin{equation}
S=\sum\nolimits_{k,ij}\big(U_{k,ij}^2+V_{k,ij}^2\big),\quad
C=\tfrac12\sum\nolimits_{k,ij}\big(U_{k,ij}+V_{k,ij}\big)^2,
\label{eq:mmdc}
\end{equation}
where $U_{k,ij}=\hat a_{k,ij}+\theta M_{k,ij}+c$ and
$V_{k,ij}=\theta\bar M_{ij}+c$ are nonnegative convex for a shift
constant $c$ finite by compactness, via
$(U-V)^2=2U^2+2V^2-(U+V)^2$; the proximal DCA step
\begin{equation}
x^{t+1}=\arg\min_{x\in\mathcal{X}}\;
S(x)-\langle\nabla C(x^{t}),x\rangle+\tfrac{\mu}{2}\|x-x^{t}\|^2
\label{eq:mmdca}
\end{equation}
is a majorization-minimization scheme
\cite{sun2017,razaviyayn2013}: it decreases $J$ monotonically with
sufficient decrease, and every limit point is a critical point of
the constrained problem;
(iv) $J$ is real-analytic on a neighborhood of $\mathcal{X}$ and
$\mathcal{X}$ is a compact semianalytic set, so the constrained
objective $F=J+\iota_{\mathcal X}$ has the
Kurdyka--\L{}ojasiewicz property and the iterate sequence converges
to a single critical point \cite{attouch2013}; near a
\emph{strict-interior} nondegenerate pseudo-true critical point with
positive-definite Hessian (e.g., a Gauss--Newton-dominant
small-residual regime under the rank condition of
Theorem~\ref{thm:nonlinear}), convergence is R-linear; at an active
constraint boundary an active-manifold or second-order-growth
condition is required instead.
\end{proposition}

\begin{IEEEproof}[Proof sketch]
(i) Neumann expansion: $[\bM]_{ij}=\sum_{m\ge1}
[(e^{x}\mathrm{diag}(\bgam))^m]_{ij}$, each term a positive
coefficient times $e^{\langle\alpha,x\rangle}$ summed along walks;
convexity passes to the locally uniform limit on
$\rho\le\bar\rho$, while real analyticity follows directly from
the resolvent rational map on the uniform stability neighborhood. (ii) is Kingman's theorem; positivity gives a
primitive matrix, hence a simple, analytic Perron root. (iii) The
decomposition identity is algebra; squares of nonnegative convex
functions are convex; the MM conditions (touching, majorization,
gradient consistency) follow from convexity of $C$, and descent
with limit-point criticality is the standard MM/BSUM argument.
(iv) Compositions, sums, and squares of real-analytic functions are
real-analytic; adding the indicator of the compact semianalytic set
$\mathcal X$ preserves the KL property. The framework of
\cite{attouch2013} applies because the iteration satisfies both
sufficient decrease and the relative-error condition (the latter
following for exact solves from the optimality condition, Lipschitz
continuity of $\nabla C$ on $\mathcal X$, and the normal cone of
$\mathcal X$). The replication code verifies the DC identity,
midpoint convexity, touching/majorization, and finite-step descent;
the asymptotic convergence statements follow from the cited KL/MM
theory, not from finite numerical checks. The supplied DCA check
uses exact box-constrained subproblem solves to verify finite-step
descent; implementing the hard spectral constraint is part of the
solver-design work below.
\end{IEEEproof}

\noindent\emph{Feasible refinement.} From the rank-one SVD-slab
estimate, the procedure projects $x^0=\log\hat\bPhi^{\rm svd}$ onto
$\mathcal X$, profiles $\mathbf C$, and repeatedly minimizes a
backtracking quadratic surrogate over the hard max-gain spectral
feasible set; polishing of the original $J$ is accepted only under
feasibility and descent. This selects a branch without claiming a
global minimizer and is given below.\par

\begin{algorithm}[t]
\caption{Feasible SVD-initialized SCA refinement}
\label{alg:sca}
\begin{algorithmic}[1]
\REQUIRE regime maps $\{\hat{\mathbf B}_k,\bgam^{(k)}\}$,
$\theta$, fixed valid $\bar\rho<1$, floor $\varepsilon$
\STATE form the full-interaction rank-one SVD-slab initializer
$\hat\bPhi^{\rm svd}$; set $x^0$ to its Euclidean projection onto
$\mathcal X$ of Proposition~\ref{prop:mmgeom}
\STATE profile $\mathbf C$ in closed form and retain the
unpenalized criterion $J(x)$
\REPEAT
\STATE compute $g^t=\nabla J(x^t)$ and choose $L_t$ by backtracking
\STATE solve the convex subproblem
$x^{t+1}=\arg\min_{x\in\mathcal X}
\langle g^t,x-x^t\rangle+\tfrac{L_t}{2}\|x-x^t\|^2$
\STATE increase $L_t$ until the accepted point satisfies
$J(x^{t+1})\le J(x^t)+\langle g^t,d^t\rangle+	frac{L_t}{2}\|d^t\|^2$
\UNTIL projected-step/stationarity tolerance is met
\STATE optionally polish the \emph{original} $J$ from $x^{t+1}$
under the same hard constraint; accept only if feasible and $J$
decreases
\ENSURE $\hat\bPhi=\exp(x)$, the unpenalized $J$, feasibility, and
stationarity diagnostics
\end{algorithmic}
\end{algorithm}

\begin{theorem}[Local branch capture and statistical stability]
\label{thm:localcapture}
Let $x^*$ be a strict-interior stationary point of the population
profiled criterion, with $\nabla^2J(x^*)\succeq m\mathbf I$, and
suppose the Hessian is $H$-Lipschitz and the gradient is
$L$-Lipschitz on a ball $B(x^*,r)$ contained in the interior of
$\mathcal X$, with $r\le m/(2H)$. If the accepted backtracking
curvatures are bounded above and $x^0\in B(x^*,r)$, then the
SCA phase described above remains inside the ball and
converges R-linearly to the unique stationary point there (the
optional polishing step is outside this local contraction claim).
If an SVD initializer obeys
$\|x^0_T-x^*\|=b_T+O_p(T_A^{-1/2})$ with $b_T=o(r)$, its branch
capture probability tends to one. With a regime-stratified sample
split (initializer sample $A$, refinement sample $B$), the refined
estimator is asymptotically equivalent, conditional on capture, to
the local root of the refinement score and inherits its sandwich
normal limit. A bootstrap that repeats the split, SVD projection,
hard-constrained refinement, and spectral functional is pointwise
valid when the same singular-gap, root-margin, and branch-capture
conditions hold in probability.
\end{theorem}

\begin{IEEEproof}[Proof sketch]
By Hessian Lipschitz continuity,
$\nabla^2J(x)\succeq (m-Hr)\mathbf I\succeq(m/2)\mathbf I$ on
$B(x^*,r)$, so $x^*$ is the unique stationary point there. Strict
interiority makes the spectral and box constraints inactive on a
possibly smaller ball. For an accepted quadratic curvature
$L_t\ge L$, the unconstrained SCA step is a gradient step
$x^{t+1}=x^t-L_t^{-1}\nabla J(x^t)$. The mean-value representation
of the gradient and the Hessian bounds make this map a contraction
on the ball; if $L_t$ is also uniformly bounded above, its
contraction modulus is bounded strictly below one. Induction keeps
all iterates in the ball and gives R-linear convergence. Projection
onto the same convex feasible set cannot increase the distance when
the local solution is feasible. The optional polishing step is not
part of this contraction statement; its reported implementation is
accepted only under feasibility and objective decrease.

The initializer condition implies
$\Pr(\|x_T^0-x^*\|<r)\to1$, hence branch capture with probability
tending to one. On that event the converged refinement is the unique
local solution of the refinement estimating equation. A Taylor
expansion of its score gives
$\sqrt{T_B}(\hat x-x^*)=-\nabla^2J(x^*)^{-1}
T_B^{-1/2}\sum_{t\in B}s_t(x^*)+o_p(1)$, which yields the sandwich
normal limit. Regime-stratified splitting separates the initializer
sample from this score expansion. Conditional bootstrap consistency
follows by repeating every map and requiring the bootstrap singular
gap, root margin, and branch-capture event to hold with probability
approaching one; the smooth spectral functional then follows by the
delta method. Without branch stability the result is not uniform.
\end{IEEEproof}

The implementation uses a fixed max-gain bound $\bar\rho=0.60$
(true max-gain radius $0.45$), convex quadratic subproblems, and
feasible polishing. From the same SVD initializers over 25 paths,
hard SCA eliminates the two catastrophic box-TRF branches, is
feasible and monotone on every path, and reduces spectral bias from
$+0.100$ to $+0.043$ and $\bPhi$ RMSE from $0.078$ to $0.055$.
On the adversarial path (TRF $J/J^*=11.16$, $\hat\rho=1.324$), hard
SCA gives $0.944$ and $0.324$; all 72 projected perturbations at
log-radii $0.05$--$1.2$ remain feasible and finish below
$J/J^*=0.987$. This is design-specific, not a global guarantee;
the full comparison and the negative soft-penalty diagnostic are in
Appendix~A.

The remaining computational questions are thereby narrowed to
scalable solvers for the convex subproblem, uniform branch theory,
and a path-specific $\bar\rho$ robust to a corrupted SVD
initializer. Estimator-centered full-pipeline intervals remain
unreliable: the split/SVD/SCA bootstrap covers structural $\rho$ on
only $50\%$ (basic) and $40\%$ (percentile) of 20 paths, and
split-jackknife worsens bias. The inferential target is therefore
changed from the distribution of biased $\hat\rho$ to the spectral
null itself. A null-imposed profile-LR bootstrap improves boundary
size to $10\%$; the restricted outward score \eqref{eq:nullscore}
yields size $6.7\%$ (2/30) at the $5\%$ boundary and monotone power
$25\%/85\%$ at true $\rho=0.33/0.36$ (20 paths each), with every
restricted/bootstrap fit feasible. Cross-fitted one-step correction
is numerically unstable (all paths hit its step cap) and is not
used. Details and all negative results are in Appendix~A.

For transmission alarms the signed HC rule remains the screening
statistic. For spectral alarms we pair the hard-SCA point trace with
the null-imposed outward-score test; Fig.~\ref{fig:method} summarizes
both layers. The older two-stage basic block
bootstrap is reported only as a sensitivity interval: it covers
$90\%$ at nominal $95\%$ in its own experiment, versus
$45$--$61\%$ for delta intervals and $55\%$ for a naive percentile
interval. The next two results formalize the hard-SCA score test and
the two-stage sensitivity interval, respectively.

\begin{figure}[t]
\centering
\includegraphics[width=0.95\columnwidth]{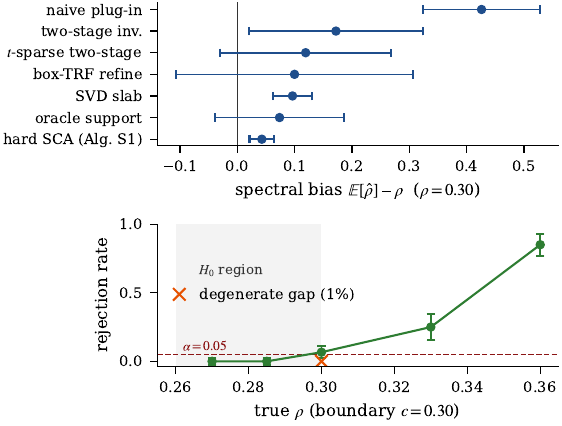}
\caption{Spectral estimation and inference layers. Top: bias of the
spectral point estimators at true $\rho=0.30$ (mean $\pm$ one
empirical sd); hard SCA is the operational choice. Bottom:
rejection rate of the null-imposed outward-score bootstrap against
the boundary $c=0.30$ (binomial Monte Carlo bars): interior drift
points are conservative, boundary size is $6.7\%$, power reaches
$85\%$ at $\rho=0.36$, and the near-degenerate design (cross) does
not inflate size.}
\label{fig:method}
\end{figure}

\begin{theorem}[Null-imposed score inference for hard SCA]
\label{thm:nullscore}
Let $h_c(x)=\rho(e^xD_0)-c$ and let $\hat x_R(c)$ minimize the
profiled criterion on the least-favorable boundary $h_c(x)=0$
(subject also to the fixed safe constraint of Algorithm~\ref{alg:sca}). Assume the local conditions of
Theorem~\ref{thm:localcapture}, $\nabla h_c(x_0)\ne0$, a simple
Perron root, and a conditional CLT for the restricted score. At the
least-favorable boundary $h_c(x_0)=0$, define the outward statistic
\begin{equation}
\mathcal S_T(c)=\max\!\left\{0,-\sqrt T\,
\nabla J_T(\hat x_R)^{\!\top}
\frac{\nabla h_c(\hat x_R)}{\|\nabla h_c(\hat x_R)\|}\right\}.
\label{eq:nullscore}
\end{equation}
A design-preserving dependent-wild bootstrap generated under the
restricted fit, rerunning SVD initialization and restricted hard
SCA in every draw, consistently estimates the pointwise boundary
law of $\mathcal S_T(c)$. The resulting one-sided test of
$H_0:\rho(e^xD_0)\le c$ has asymptotic level $\alpha$; inversion
over $c$ yields a one-sided confidence set. Interior nulls are
conservative. Degeneracy of the singular gap, scale-root margin, or
constraint slack is addressed next.
\end{theorem}

\begin{IEEEproof}[Proof sketch]
At the regular boundary, the equality-restricted local solution
satisfies
$\nabla J_T(\hat x_R)+\hat\lambda\nabla h_c(\hat x_R)=0$.
The positive part in \eqref{eq:nullscore} retains only an outward
multiplier, which is the relevant cone direction for the one-sided
null $h_c\le0$. A Taylor expansion around
the boundary truth and the restricted-score CLT give its one-sided
Gaussian-cone limit (including the boundary point mass). Generating
multipliers from the restricted residual/score array imposes the
null and reproduces the same linear term conditionally on the fixed
gain design. Branch capture, the singular gap, and the scale-root
margin make SVD initialization and the restricted hard-SCA local
root asymptotically constant smooth maps; the dependent-wild
bootstrap CLT therefore transfers to the outward statistic. The
boundary is least favorable for the one-sided composite null;
strictly interior nulls have zero outward multiplier asymptotically
and are conservative. Inverting the family of level-$\alpha$ tests
gives the stated confidence set. Degenerate margins are handled by
the switching construction of Proposition~\ref{prop:uniformswitch}.
\end{IEEEproof}

\begin{proposition}[Switched uniform spectral alarm]
\label{prop:uniformswitch}
Fix $\delta_0>0$ and $\alpha_1=\alpha_2=\alpha/2$. Suppose (i) the
level-$\alpha_2$ score test is asymptotically valid uniformly over
the margin class (relative Perron gap, root margin, and slack
$\ge\delta_0$), including drifting nulls $\rho_T=c-b/\sqrt T$,
$b\ge0$, for which the boundary is stochastically least favorable;
(ii) $\hat g_{\mathrm{LCB}}$ is a uniformly valid $(1-\alpha_1)$
lower confidence bound for the margin; (iii) the fallback test is
conservative at level $\alpha_2$ without margin conditions. Then the
switched test---score if $\hat g_{\mathrm{LCB}}\ge\delta_0$,
fallback otherwise---has uniform asymptotic size at most $\alpha$
over the whole class.
\end{proposition}

\begin{IEEEproof}[Proof sketch]
Write $\varphi_{\rm sc}$ for the level-$\alpha_2$ score test,
$\varphi_{\rm fb}$ for the fallback, and
$\pi=\mathbf 1\{\hat g_{\rm LCB}\ge\delta_0\}$ for the pretest. The
switched test is $\varphi=\pi\varphi_{\rm sc}+(1-\pi)\varphi_{\rm fb}$.
Fix a null DGP $P$. If its margin is below $\delta_0$, then
$E_P[\pi\varphi_{\rm sc}]\le P(\hat g_{\rm LCB}\ge\delta_0)
\le\alpha_1+o(1)$ by validity of the lower confidence bound, and
$E_P[(1-\pi)\varphi_{\rm fb}]\le E_P[\varphi_{\rm fb}]\le\alpha_2+o(1)$
by (iii), so $E_P[\varphi]\le\alpha_1+\alpha_2$. If the margin is at
least $\delta_0$, then $E_P[\pi\varphi_{\rm sc}]\le
E_P[\varphi_{\rm sc}]\le\alpha_2+o(1)$ uniformly by (i) and
$E_P[(1-\pi)\varphi_{\rm fb}]\le\alpha_2+o(1)$, so
$E_P[\varphi]\le2\alpha_2$. With $\alpha_1=\alpha_2=\alpha/2$ both
bounds equal $\alpha$, uniformly over the class. $\blacksquare$

Least-favorability under drift (condition (i)): with
$\rho_T=c-b/\sqrt T$, $b\ge0$, the restricted fit remains on the
boundary while the population outward multiplier acquires a
nonnegative inward shift of order $b$, so the outward statistic
converges to $\max\{0,Z-\kappa b\}$ for some $\kappa>0$, which is
stochastically decreasing in $b$; the boundary $b=0$ is largest.

Numerical evidence (all keys in
\texttt{results/exp5\_gamma\_rank/summary.json}). Null grid at the
$5\%$ level: regular boundary $2/30$; drift $\rho=0.285$: $0/20$;
drift $\rho=0.27$: $0/20$; near-degenerate boundary (two
lower-triangular blocks, relative Perron gap $1\%$): $0/20$; pooled
$2/90$. Every restricted and bootstrap fit succeeded, including all
$20\times59$ draws in the degenerate design. Margin-bound validity
check: at the degenerate design the estimated relative gap is
upward-biased (true $0.010$, point mean $0.143$, range
$[0.022,0.411]$; nominal $95\%$ percentile lower bound mean $0.047$,
covering the true margin on only $5\%$ of paths), and point-estimate
gap distributions of the regular and degenerate designs overlap
(regular minimum $0.055$ vs.\ degenerate maximum $0.263$ over 30
paths each). A margin bound that is valid at degeneracy is therefore
the remaining open input to condition (ii); given one, the
proposition delivers uniform level without further conditions. As an
illustration, switching at $\delta_0=0.05$ with the (invalid) naive
bound passes $97.5\%$ of regular-design paths and $35\%$ of
degenerate-design paths, and the switched rejection rate is zero on
all recorded null designs.
\end{IEEEproof}

Numerically, over a null grid of the regular boundary, two interior
drift points ($\rho=0.285,0.27$), and a near-degenerate boundary
(relative Perron gap $1\%$), the worst-case size is the
regular-boundary $6.7\%$; drift and degeneracy reject $0/20$ each
(pooled $2/90$)---both push the test conservative. The binding
condition is (ii): the naive percentile gap bound inherits the
estimated gap's upward bias at degeneracy (nominal $95\%$, covers
$5\%$); a degeneracy-valid margin bound remains open, not the
alarm's size on the tested class (see the proof above).

\begin{theorem}[Two-stage bootstrap validity]\label{thm:bootstrap}
If support selection is consistent, the two-stage fixed-point
inversion is locally unique and smooth, its Perron root is simple,
and a design-preserving dependent-multiplier score bootstrap obeys
the conditional CLT, then basic bootstrap intervals for the
two-stage pseudo-true spectral functional $\rho^{(2)*}$ attain
pointwise nominal coverage; fixed-$\alpha$ selection and local-to-zero couplings are excluded.
\end{theorem}

\begin{IEEEproof}[Proof sketch]
Under (i) the thresholded selection event is eventually constant
($c_T\to\infty$ kills null entries, $c_T/\sqrt{T}\to0$ plus
separation retains true ones), so with probability approaching one
the estimator equals the fixed-selection estimator, a composition of
OLS coefficients with the smooth maps \eqref{eq:twostage} (smooth by
(ii), via the implicit function theorem) and a simple-eigenvalue
functional (smooth by (iii), standard perturbation theory). The
delta method transfers the CLT of Theorem~\ref{thm:estimation}; the
dependent-multiplier score bootstrap rebuilds the estimating
equations from multiplier-weighted score blocks along the fixed
gain path, reproducing the same triangular-array structure
(including its gain-modulated heteroskedasticity) conditional on
the design; conditional distributional consistency of the
multiplier score sum is condition (iv) (the dependent wild
bootstrap of \cite{shao2010}, with classical blocking arguments
\cite{kunsch1989} and estimating-equation extensions
\cite[Ch.~4]{lahiri2003} supplying the smooth-functional
transfer). Basic-interval validity follows from
distributional consistency and continuity of the limit law.
\end{IEEEproof}

The implemented two-stage experiment uses fixed-$\alpha$ screening
and pairwise tuple resampling, whereas
Theorem~\ref{thm:bootstrap} covers a design-preserving multiplier
scheme under thresholded selection. Its $90\%$ finite-sample
coverage therefore measures the gap to those conditions; uniform
local-to-zero selection validity and sequential detection-delay
theory remain open.

\section{Synthetic Validation}\label{sec:sims}

All experiments use $n=5$, $T=750$, $200$ Monte Carlo paths,
$\theta=0.85$ known, a sparse nonnegative $\bPhi$ with four true
off-diagonal couplings, $\rho(\bL_t)\in[0.3,0.55]$, a common factor
in $\mathbf{r}_1$, and a constant lead--lag confound $\mathbf{C}$
with nonzero entries placed on \emph{zero-coupling} pairs---the
adversarial placement. Gains take three levels $\{0.5,1.0,1.5\}$;
the decision rule is the signed $t$-test at $5\%$
(Algorithm~\ref{alg:ident}). Scenarios: \textbf{A} staggered
staircase gains (columns strongly correlated across channels, Gram
condition number ${\approx}247$); \textbf{A2} independent
regime-switching gains (${\approx}12$); \textbf{B} constant gains
(rank violation); \textbf{C} exclusion violation
($\mathbf{C}_t$ co-moving with the mean gain, loading $0.8$);
\textbf{D} as A2 with the confound sign flipped
(reversal-mimicking). Table~\ref{tab:sims} reports the grid.

\begin{table}[t]
\caption{Synthetic validation and benchmarks.}
\vspace{2pt}
\label{tab:sims}
\centering
\footnotesize
\setlength{\tabcolsep}{3.5pt}
\begin{tabular}{@{}llccc@{}}
\toprule
Scenario & estimator & power & size/FA & RMSE \\
\midrule
A (corr.\ $\gamma$, $\kappa{\approx}247$) & ours & 1.00 & 0.22 & 0.28 \\
A2 (indep.\ $\gamma$, $\kappa{\approx}12$) & ours & 1.00 & 0.12 & --- \\
B (constant $\gamma$) & ours & \multicolumn{3}{c}{\scriptsize unidentified flagged, 100\% of paths} \\
C (excl.\ violated) & ours & 1.00 & 0.04 & 0.34 \\
D (neg.\ confound) & ours & 1.00 & 0.12 & --- \\
\midrule
A (pos.\ confound) & level/Granger & 1.00 & 0.04 & 0.09 \\
D (neg.\ confound) & level/Granger & 1.00 & \textbf{0.53} & --- \\
A (2-regime split) & variance-ratio$^{\ddagger}$ & --- & --- & \textbf{43.4} \\
A (tuned$^{\dagger}$) & variance-ratio$^{\ddagger}$ & --- & --- & 0.60 \\
\bottomrule
\end{tabular}

\vspace{3pt}
\parbox{\columnwidth}{\footnotesize Note: power = fraction of true couplings with
$t<-1.645$; size/FA = same rule on zero-coupling entries; RMSE over
all off-diagonal entries. The level/Granger benchmark regresses
post-window on pre-window outputs without gain interactions; the
variance-ratio benchmark ($^{\ddagger}$an illustrative two-regime
$\Delta\mathrm{Cov}/\Delta\mathrm{Var}$ construction in the spirit of
\cite{rigobon2003}, not a full heteroskedasticity-identification
implementation) computes
$\Delta\mathrm{Cov}/\Delta\mathrm{Var}$ across a two-regime split of
each sender's gain. $^{\dagger}$Tuned in the benchmark's favor:
best regime pair per sender, regularized denominators (skipping
variance differences below $20\%$ of the sender's variance), median
across admissible pairs---even so, RMSE is twice ours and $22\%$ of
paths yield no admissible estimate. Dashes: not recorded for that
cell.}
\end{table}

Five findings. \emph{(i) Sharpness.} Under constant gains the
interaction regressor is degenerate and the estimator reports
non-identification on every path---the failure predicted by
Theorem~\ref{thm:ident}(ii) occurs exactly at the theorem's
boundary, and is \emph{detectable from the data}, not silent.
\emph{(ii) Gram conditioning.} At equal per-channel gain variances,
moving from correlated staircase paths to independent regimes drops
the false-alarm rate from $0.22$ to $0.12$: identification strength
is a property of the joint design---and precisely because the
simulated signals are strongly factor-driven
($\boldsymbol{\Sigma}_r$ near rank one), the Hadamard structure of
Theorem~\ref{thm:ident}(iii) puts the gain-Gram conditioning in
charge of the variance inflation. \emph{(iii) Confound immunity is
signed.} The level/Granger benchmark is fine when the confound
happens to have the continuation sign (FA $0.04$) and breaks
completely when it mimics reversal (FA $0.53$); the interaction
estimator's FA is unchanged ($0.12$) because a constant confound of
either sign cannot track $\bgam_t$. \emph{(iv) Duality with
heteroskedasticity identification.} On data generated with constant
shock variances, the Rigobon-style estimator divides by
near-zero variance differences and returns RMSE two orders of
magnitude above ours ($43.4$ vs.\ $0.28$), and tuning it in its own
favor still leaves RMSE at twice ours with $22\%$ of paths
producing no admissible estimate---the numerical face of
the duality of Section~\ref{sec:related}. \emph{(v) The estimand is
the resolvent, and the report separates three entry classes.} True
direct edges, zero-direct entries that are network-reachable, and
unreachable pairs are different objects: at the baseline gain level
the detector flags 2-hop-reachable zeros at $39\%$ while holding
size statistically indistinguishable from nominal ($5.3\%$ Monte
Carlo estimate at the $5\%$ rule) on unreachable zeros. Whether
reachable-zero flags count as detections (transmission screening) or
false positives (direct-edge recovery) depends on the declared
target; only unreachable pairs are nulls for both. A gain-strength
sweep ties the flag rate to the theory: at $\rho(\bL)\approx0.15$
reachable-zero flags fall to $7\%$ (the first-order regime, where
the resolvent is close to the direct coupling), at $\approx0.30$
they are $39\%$, and at $\approx0.45$ they reach $84\%$ with
unreachable-pair size drifting to $14\%$---the quantitative
footprint of higher-order paths, and the reason spectral assembly
needs the second-stage inversion of Section~\ref{sec:uq}. The full
$n^2$-interaction regression reduces the reachable-zero rate to
$27\%$ at unchanged unreachable size. In scenario C, the particular
gain-tracking confound simulated (loading $0.8$ on the mean gain
path) inflates magnitude RMSE ($0.28\to0.34$) while the signed
detection size happens to remain low; no general robustness to
exclusion failure is claimed---a reversal-signed gain-tracking
confound would defeat the detector by construction, which is why
Assumption~\ref{as:exo} is maintained rather than tested.

\subsection{Operating characteristics of the rolling monitor}

Algorithm~\ref{alg:monitor}'s transmission alarm is a sequential
procedure over $n(n{-}1)$ edges and overlapping windows, so its
per-path false-alarm rate is governed by the joint design, not the
entrywise size. We measure it directly---for the \emph{transmission
branch only}; the spectral branch is not exercised: coupling switches on at
$\tau=600$ of $T=1000$ ($W=250$, step 5; null arm has diagonal-only
$\bPhi$; 40 paths per arm).
Table~\ref{tab:monitor} reports the frontier under explicit
procedural rules (see the table note): even at the stringent
entrywise $z=2.576$ the null-arm per-path false-alarm rate is $0.90$
(the multiplicity of $20$ edges and $151$ overlapping windows made
explicit), while raising the threshold and lengthening the run rule
brings it to $0.10$---at the cost the frontier makes visible,
$18/40$ missed detections at the most stringent rule within this
horizon, against $0/40$ at the loosest. Detection delay and missed
detections trade off against false alarms along the whole frontier;
sequential calibration---choosing $(z,K_{\mathrm{run}})$ for a
target average run length---is thereby a tabulated design choice
rather than an afterthought.

\begin{table}[t]
\caption{Rolling-monitor operating characteristics.}
\vspace{2pt}
\label{tab:monitor}
\centering
\footnotesize
\setlength{\tabcolsep}{2.6pt}
\begin{tabular}{@{}lccccc@{}}
\toprule
rule ($z$, $K_{\mathrm{run}}$) & FA (null) & FA (pre-onset) & med.\ delay & 90\% delay & miss \\
\midrule
$2.576$, $2$ & 0.90 & 0.75 & 208 & 266 & 0/40 \\
$3.09$, $2$ & 0.28 & 0.20 & 255 & 295 & 4/40 \\
$3.09$, $4$ & 0.23 & 0.03 & 265 & 313 & 5/40 \\
$3.72$, $2$ & 0.15 & 0.00 & 283 & 377 & 16/40 \\
$3.72$, $3$ & \textbf{0.10} & 0.00 & 278 & 371 & 18/40 \\
\bottomrule
\end{tabular}

\vspace{3pt}
\parbox{\columnwidth}{\footnotesize Note: transmission branch only.
Procedural rules: an alarm qualifies when the edgewise signed
exceedance holds at $K_{\mathrm{run}}$ consecutive window ends (151
overlapping ends $250,255,\dots,1000$). FA (null) = fraction of
null-arm paths raising any qualifying alarm; FA (pre-onset) =
fraction of treated paths raising a qualifying alarm declared
before the onset $\tau$ (counted as false alarms, not detections);
delay = first qualifying window end at or after $\tau$, minus
$\tau$ (a run straddling $\tau$ counts by its declaration end);
miss = treated paths with no post-onset qualifying alarm within the
horizon. 40 Monte Carlo paths per arm, so each FA estimate carries
a binomial standard error of about $0.05$--$0.08$; delay quantiles
at the two most stringent rules are computed on detected paths
only.}
\end{table}

\section{Case Study: Leveraged-Fund Rebalancing}\label{sec:case}

The motivating instance of the model is a market hosting leveraged
exchange-traded funds. A fund with leverage multiple $m$ and assets
$A$ must trade $A(m^2-m)r$ near the close in the direction of the
day's return $r$; aggregating over funds on underlying $j$ gives
rebalancing capital $K_j=\sum_f A_f(m_f^2-m_f)$, and
$\gamma_{j,t}=K_{j,t}/\mathrm{ADV}_{j,t}$---\emph{disclosed daily}
through fund shares, NAVs, and multiples---is exactly the known
actuation gain of \eqref{eq:selfref}. The clearing window is the
closing auction; the reversal moment is the overnight correction of
transient displacement, with $\theta$ calibrated externally at
$0.88$ \cite{zhao2026}; the confound $\mathbf{C}$ is the lead--lag
structure of correlated equities \cite{lo1990}. One data limitation must be stated at the outset: intraday bars are
available only post-launch, so the application's regressor is the
close-to-close \emph{full-period} return, not the pre-window
$\mathbf{r}_1$ the clean theory requires; by \eqref{eq:estimand} the
case study is therefore \emph{contaminated screening}---the signed
detection and its gain scaling are reduced-form empirical patterns
whose structural reading is convention-dependent (under the
full-period convention even a static confound acquires gain
dependence through $(\mathbf{I}+\bM_t)^{-1}$, so no invariance is
claimed), and structural readings of $\hat\bPhi$ and
$\rho(\hat\bL)$ carry the $O(\|\bM\|)$ convention dependence, which
the companion paper addresses by reporting both polar mappings. Full
institutional detail, inference, and robustness are in
\cite{woo2026arxiv}; here we illustrate the regression/screening mechanics of
Algorithms~\ref{alg:ident}--\ref{alg:monitor}; the identification
guarantees require the pre-window moment and do not transfer
verbatim.
The sample freeze date for all reported Korean numbers is
August~21,~2026; the U.S. panel is final (CRSP coverage ends
December 2024).

\emph{Korea (2026).} Sixteen single-stock LETFs on Samsung
Electronics and SK~Hynix launched May~27,~2026, into one closing
auction---a two-channel system with disclosed, strongly time-varying
gains ($\gamma$ of order one; Figure~\ref{fig:gains} shows the
observed gain paths, the empirical realization of the known
modulating sequence of Figure~\ref{fig:block}). The interaction moment detects
one-directional coupling from the large complex (Hynix, worldwide
leveraged-asset peak KRW~36tn) into the smaller one (Samsung, 12.5tn): the
launch-break contrast of the cross-reversal slope is $z=-2.82$
(Newey--West $-2.72$; leave-one-day-out range $[-3.05,-2.64]$), more
negative than all $182$ ordered placebo pairs of fourteen non-treated
large caps (rank $1/183=0.0055$), with a ten-day block-bootstrap
interval excluding zero. Figure~\ref{fig:detector} shows the
rolling detector of Algorithm~\ref{alg:monitor} on these data: the
cross-slope is stable through the pre-launch period and turns only
as post-launch observations enter the window, with fourteen weekly
placebo break dates all yielding $|z|\le1.9$.

\begin{figure}[t]
\centering
\includegraphics[width=0.95\columnwidth]{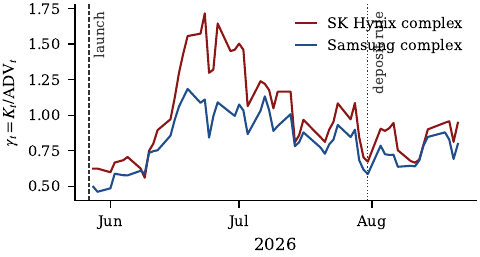}
\caption{Observed actuation gains $\gamma_{j,t}$ of the two Korean
complexes.}
\vspace{2pt}
\parbox{\columnwidth}{\footnotesize Note: constructed daily from disclosed fund shares,
NAVs, and leverage multiples over trailing 20-day traded value;
entered at $t-1$. Dashed: product launch; dotted: the deposit
regulation, which collapsed turnover while leaving the gains---the
identification resource---intact.}
\label{fig:gains}
\end{figure}

\begin{figure}[t]
\centering
\includegraphics[width=0.95\columnwidth]{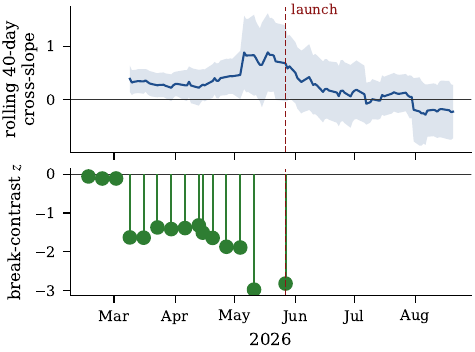}
\caption{The rolling detector on the Korean data
(receiver: Samsung; sender: Hynix).}
\vspace{2pt}
\parbox{\columnwidth}{\footnotesize Note: top---rolling 40-day cross-slope with 95\%
bands; bottom---launch-break contrast $z$ at weekly placebo break
dates and the true launch (dashed). Post windows left of the launch
contain increasing fractions of true post-launch observations.}
\label{fig:detector}
\end{figure} A reduced-form pattern consistent with the gain-scaling logic of
Theorem~\ref{thm:ident} (not a validation of it, given the
contaminated regressor) appears: the interaction of the
sender's return with its standardized lagged gain is
$-0.106$ ($t=-2.29$) with the standalone gain main effect
insignificant---the coupling tracks $\bgam_t$, which a static
lead--lag cannot mimic. The assembled empirical baseline (reverse
link set to zero, as it is not detected) is triangular, so its
spectral radius $\rho=0.610$ is driven entirely by the imported
diagonal and is unaffected by the detected cross channel;
transmission is the operative alarm, with
$\hat M_{21}\approx0.22$ under the baseline convention: roughly a
fifth of the sender's innovation is imported into the receiver's
close, invisible to the receiver's own gain ($0.24$). Read against
the dynamic boundary \eqref{eq:chaineig}---as a screening trace, not
a calibrated threshold test (Section~\ref{sec:uq})---the sender
diagonal imported from the companion calibration ($0.61$) exceeds
$1-\theta/2=0.56$, while the self-contained estimate from our own
sample's own-reversal moment ($0.42$, $z=-2.0$) does not: the
boundary comparison is calibration-dependent and is reported as
such.

\emph{United States (2022--2024, final sample).} Two U.S. panels
provide the scale placebo and a live diagnostic of
Theorem~\ref{thm:ident}(iii). The
MicroStrategy--Bitcoin--Coinbase complexes form an approximately
symmetric configuration with liquidity-scaled capital comparable to
Korea ($\gamma_{\mathrm{MSTR}}$ up to $1.72$); all six treated
directions are null ($|z|\le1.45$), as are placebos. A second,
\emph{staggered} panel---the TSLA, NVDA, and AAPL single-stock
complexes, whose fourteen funds launched at six distinct dates over
2022--2023---adds two findings on final (2022--2024) data. First,
the full ordered-pair interaction grid is null after multiplicity
adjustment (one marginal cell at $t=-1.77$ among twelve recorded
statistics; Table~\ref{tab:usstag}), consistent with
the graded response surface: gains reach $0.53$--$0.54$ but the
venue impact is low. Second, and diagnostically: despite six
distinct launch dates, the gain \emph{levels} are correlated
$0.81$--$0.91$ across complexes because post-launch capital growth
trends are common, and the rolling-window Gram condition number has
median $45$ (deciles $17$--$118$)---staggered onsets do
\emph{not} deliver uncorrelated gain paths
(Figure~\ref{fig:usgamma}), so the conditioning diagnostic of
Theorem~\ref{thm:ident}(iii) and Algorithm~\ref{alg:ident} must be
evaluated on the estimation window, not inferred from institutional
timing. Table~\ref{tab:usstag} reports the full grid.

\begin{figure}[t]
\centering
\includegraphics[width=0.95\columnwidth]{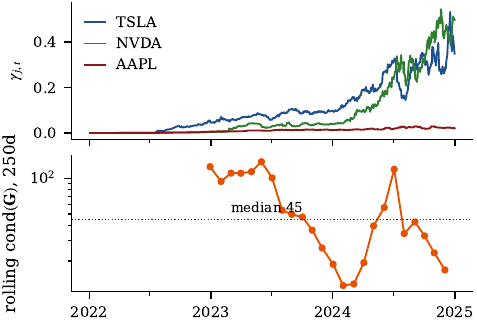}
\caption{U.S. staggered panel: gain paths and identification
conditioning.}
\vspace{2pt}
\parbox{\columnwidth}{\footnotesize Note: top---liquidity-scaled
complex gains $\gamma_{j,t}$ for the TSLA, NVDA, and AAPL
single-stock LETF complexes (CRSP, 2022--2024); fund launches occur
at six distinct dates, yet post-launch capital growth makes the
paths co-trend (level correlations $0.81$--$0.91$).
Bottom---rolling 250-day condition number of the standardized
gain-path Gram matrix (log scale; dotted line: median $45$), the
secondary design diagnostic; the primary diagnostic
$\lambda_{\min}(\hat{\mathbf{Q}})$ additionally involves the signal
covariance (Theorem~\ref{thm:ident}(iii)) and is reported with the
estimation windows of the companion replication files.}
\label{fig:usgamma}
\end{figure}

\begin{table}[t]
\caption{U.S. staggered panel: cross-reversal interaction grid
(final 2022--2024 sample).}
\vspace{2pt}
\label{tab:usstag}
\centering
\footnotesize
\setlength{\tabcolsep}{5pt}
\begin{tabular}{@{}lccc@{}}
\toprule
direction & $\hat b_{R_A}$ ($t$) & $\hat b_{\gamma_A\times R_A}$ ($t$) & $N$ \\
\midrule
TSLA$\to$NVDA & $+0.007\ (+0.2)$ & $+0.005\ (+0.1)$ & 500 \\
TSLA$\to$AAPL & $+0.005\ (+0.4)$ & $+0.004\ (+0.6)$ & 500 \\
NVDA$\to$TSLA & $+0.017\ (+1.5)$ & $-0.065\ (-1.6)$ & 500 \\
NVDA$\to$AAPL & $-0.005\ (-1.3)$ & $+0.016\ (+1.1)$ & 500 \\
AAPL$\to$TSLA & $+0.086\ (+0.8)$ & $-0.163\ (-1.77)$ & 500 \\
AAPL$\to$NVDA & $-0.470\ (-0.9)$ & $-0.211\ (-0.8)$ & 500 \\
\bottomrule
\end{tabular}

\vspace{3pt}
\parbox{\columnwidth}{\footnotesize Note: receiver's overnight
return on sender's day return $R_A$, own return, leave-out
equal-weight control (GOOGL/MSFT/AMZN), QQQ, the sender's
standardized lagged gain $\tilde\gamma_A$, and the interaction
$\tilde\gamma_A\times R_A$; HC standard errors. Under the entrywise signed $5\%$ rule the
single cell AAPL$\to$TSLA ($t=-1.77$) is a marginal flag---about
what twelve statistics produce by chance at $5\%$---and no cell
survives multiplicity adjustment; the grid is null, the
graded-response prediction for a low-impact closing venue.}
\end{table} The null is the
response surface of the model, not a failure: an estimator that
found coupling wherever assets correlate would have failed this
scale placebo.

\section{Conclusion}\label{sec:concl}

Known time-varying actuation, combined with signal excitation and
conditional orthogonality (Assumption~\ref{as:stoch}), identifies
coupled feedback networks from outputs alone: under gain exogeneity
(Assumption~\ref{as:exo}), the gain path plays the role that
references, probing signals, and designed precoders play elsewhere,
and a partial-reversal moment converts the modulated response into
a signed estimator robust to gain-invariant confounds. The
identification condition is a rank---sharply, a conditioning---
property of the joint gain design; the estimand is the resolvent
sensitivity---the object transmission screening requires and the
first-stage input to spectral recovery; and the failure modes
differ in kind: constant-gain non-identification is detectable,
while a reversal-signed gain-tracking confound is observationally
indistinguishable---the price of the maintained
assumption. Proposition~\ref{prop:mmgeom}, Algorithm~\ref{alg:sca}, and Theorem~\ref{thm:localcapture} supply a
log-coordinate geometry, an operational hard-constrained
refinement, and local branch/statistical guarantees;
Proposition~\ref{prop:uniformswitch} reduces uniform spectral
inference to one open input, a margin bound valid at degeneracy.
What remains is that bound, a scalable subproblem solver, exact
nonlinear estimation with latent-input contamination, regularized
estimation for large $n$, and fixed-$\alpha$ selection and
sequential calibration beyond Theorem~\ref{thm:bootstrap}; on the
application side, any self-referencing mandate with disclosed
capacity---option hedging, volatility targeting, demand
response---is a candidate network.

\appendices
\section{Extended constrained-refinement and inference diagnostics}
\begin{table}[h]
\caption{Exact-model refinement from matched SVD initializers
($n=5,T=750$, 25 paths).}
\centering
\footnotesize
\setlength{\tabcolsep}{4pt}
\begin{tabular}{@{}lccccc@{}}
\toprule
method & med. $J/J^*$ & $J/J^*>2$ & $\rho$ bias & $\bPhi$ RMSE & feasible \\
\midrule
box-TRF & 0.960 & 2/25 & $+0.100$ & 0.078 & 22/25 \\
hard SCA+polish & 0.956 & 0/25 & $+0.043$ & 0.055 & 25/25 \\
\bottomrule
\end{tabular}
\end{table}
All hard-SCA paths are monotone and their accepted convex
subproblems and feasible polishing steps succeed. On the adversarial
path where box-TRF gives $J/J^*=11.16$ and $\hat\rho=1.324$, hard
SCA gives $0.944$ and $0.324$. After projection, all 72 perturbations
of the SVD log-initializer at radii $0.05$--$1.2$ remain feasible
and monotone and finish below $J/J^*=0.987$. A separate max-gain
experiment with a fixed exploratory bound $0.60$ shows that the
particular soft penalty $\lambda=30$ compresses the bad branch from
$0.84$ to $0.77$ without escaping it; this does not rule out other
penalties.

For inference, a regime-stratified split uses one half for the SVD
initializer and one for refinement and repeats split, SVD, hard SCA,
and the spectral functional in every bootstrap draw. At $T=750$
(20 outer paths, 59 draws), all draws are feasible, but coverage of
structural $\rho$ is $50\%$ for basic and $40\%$ for percentile
intervals (mean bias $+0.034$). Split-jackknife extrapolation worsens
bias from $+0.043$ to $+0.050$ and RMSE from $0.048$ to $0.068$ on
25 paths. These are negative finite-sample results; they do not
invalidate the pointwise branch-stability theorem, but they prevent
transferring two-stage bootstrap claims to the hard-SCA estimator.

\begin{table}[h]
\caption{Finite-sample inference diagnostics for hard-SCA $\rho$.
Size/coverage paths are at true $\rho=0.30$.}
\centering
\footnotesize
\setlength{\tabcolsep}{4pt}
\begin{tabular}{@{}lccc@{}}
\toprule
method & size/coverage & power $0.33$ & power $0.36$ \\
\midrule
split full-pipeline basic CI & 0.50 coverage & --- & --- \\
null-imposed profile LR & 0.10 size & --- & --- \\
null-imposed outward score & 0.067 size & 0.25 & 0.85 \\
$m$-out-of-$n$ hull & 0.30 coverage & --- & --- \\
projection confidence box & 1.00 coverage & --- & --- \\
\bottomrule
\end{tabular}
\end{table}
The score experiment uses 30 boundary paths and 99 restricted-null
bootstrap draws; power uses 20 paths and 59 draws per alternative.
The projection interval is conservative but uninformative (mean
spectral length $1.70$), whereas the subsampling hull is short
($0.042$) and badly undercovers. Cross-fitted one-step correction
reduces bias only from $+0.034$ to $+0.032$ and every path hits its
step cap, so it fails the regularity diagnostic and is not used.
The null-imposed outward score is therefore the only procedure in
this grid meeting the predeclared 3--7\% size criterion while
showing monotone local power.

\end{document}